\documentclass{elsarticle}
\usepackage{amsmath,amssymb,amsthm}
\usepackage{parskip}
\usepackage{tikz}
\usepackage[numbers]{natbib}
\usepackage[margin=2.25cm]{geometry}
\usepackage[bookmarks=false,colorlinks=true,citecolor=blue]{hyperref}

\usepackage{algorithm2e}

\usepackage{graphicx}
\usepackage{xcolor}
\usepackage{subcaption}
\usepackage{amssymb}
\usepackage{algorithm2e}
\usepackage{amsmath}

\usepackage{float}
\usepackage{hyperref}
\usepackage[noabbrev]{cleveref}

\usepackage[utf8]{inputenc}
\usepackage{amssymb}
\usepackage{amsmath}
\usepackage{xspace}
\usepackage[inline]{enumitem}

\newtheorem{theorem}{Theorem}{\bfseries}{\itshape}
\newtheorem{definition}{Definition}{\bfseries}{\itshape}
\newtheorem{lemma}{Lemma}{\bfseries}{\itshape}
{\bfseries}{\itshape}
\newtheorem{observation}{Observation}{\bfseries}{\itshape}
{\bfseries}{\itshape}
\newtheorem{remark}{Remark}{\bfseries}{\itshape}
\newtheorem{question}{Question}{\bfseries}{\itshape}
\newtheorem{property}{Property}{\bfseries}{\itshape}
\newtheorem{condition}{Condition}{\bfseries}{\itshape}
\newtheorem{conjecture}{Conjecture}{\bfseries}{\itshape}

\title{\textbf{Edge-coloured graphs with only monochromatic perfect matchings and their connection to quantum physics}}
\author{\normalsize\textbf{L. Sunil Chandran}}
\author{\textbf{Rishikesh Gajjala}}
\address{Indian Institute of Science, Bengaluru\\ Email: \texttt{sunil@iisc.ac.in, rishikeshg@iisc.ac.in}}
\date{}

\makeatletter
\def\ps@pprintTitle{%
 \let\@oddhead\@empty
 \let\@evenhead\@empty
 \def\@oddfoot{\centerline{\thepage}}%
 \let\@evenfoot\@oddfoot}
\makeatother

\begin{document}

	\begin{abstract}
		\noindent Krenn, Gu and Zeilinger initiated the study of PMValid edge-colourings because of its connection to a problem from quantum physics. A graph is defined to have a PMValid $k$-edge-colouring if it admits a $k$-edge-colouring (i.e. an edge colouring with $k$-colours) with the property that
		\begin{itemize}
		\item All perfect matchings are monochromatic
		\item Each of the $k$ colour classes contain at least one perfect matching
		\end{itemize}
The matching index of a graph $G$, $\mu(G)$ is defined as the maximum value of $k$ for which $G$ admits a PMValid $k$-edge-colouring. It is easy to see that $\mu(G)\geq 1$ if and only if $G$ has a perfect matching (due to the trivial $1$-edge-colouring which is PMValid). Bogdanov observed that for all graphs non-isomorphic to $K_4$, $\mu(G)\leq 2$ and $\mu(K_4)=3$. However, the characterisation of graphs for which $\mu(G)=1$ and $\mu(G)=2$ is not known. In this work, we answer this question. Using this characterisation, we also give a fast algorithm to compute $\mu(G)$ of a graph $G$. In view of our work, the structure of PMValid $k$-edge-colourable graphs is now fully understood for all $k$. Our characterisation, also has an implication to the aforementioned quantum physics problem. In particular, it settles a conjecture of Krenn and Gu for a sub-class of graphs. 
		
	\end{abstract}

	\maketitle

\section{Introduction}
\subsection{Problem definition}
We will be considering only simple, finite, undirected and loopless graphs. In an edge-colouring of $G$, each edge $e$ is associated with a colour which is denoted using a natural number. We emphasise that this need not be a proper edge-colouring. If exactly $k$ colours are used in an edge-colouring, we call it to be a $k$-edge-colouring. $G_c$ represents an edge-coloured graph formed by a $k$-edge-colouring $c$ over graph $G$.


On a graph $G$, a $k$-edge-colouring $c$ is said to be \textit{PMValid} if

\begin{itemize}
		\item All perfect matchings are monochromatic
		\item Each of the $k$ colour classes contain at least one perfect matching
\end{itemize}
Given a graph $G$, its \textit{matching index} $\mu(G)$ is defined as the maximum value of $k$ for which $G$ admits a PMValid $k$-edge-colouring. 

\begin{question}
\label{graph_question}
For each $i \in \mathbb{N}$, can one characterise all graphs for which $\mu(G)=i$?
\end{question}
 
Bogdanov observed the following \cite{bogdanov,Quantum_graphs}
\begin{theorem}
\label{bogdanov}
For a graph $G$ which is non-isomorphic to $K_4$, $\mu(G)\leq 2$ and $\mu(K_4)=3$.
\end{theorem}

\subsection{Motivation}

If the reader of the paper is not at all comfortable with quantum physics, they may just casually browse through this sub-section since the material of this paper can be essentially understood from a purely graph-theoretic point of view. But we emphasise that the graph theoretic problem here is strongly motivated by an earlier work \cite{Quantum_graphs}, the authors of which are prominent quantum physicists, one of them being a Nobel laureate in physics of 2022. 

Quantum entanglement theory implies that two particles can influence each other, even though they are separated over large distances. In 1964, Bell demonstrated that quantum mechanics conflicts with our classical understanding of the world \citep{bell}. Later, in 1989, Greenberger, Horne, and Zeilinger (abbreviated as GHZ) studied what can happen if more than two particles are entangled \cite{Greenberger}. Such states in which three particles are entangled (
$|GHZ_{3,2}\rangle = \frac{1}{\sqrt{2}}\left(|000\rangle + |111\rangle \right)$) were observed rejecting local-realistic theories  ~\cite{PhysRevLett.82.1345,Pan2000}. We note that the work on experimentally constructing GHZ states is at the heart of the topic for which Zeilinger was a co-recipient of the Nobel prize in 2022 \cite{nobel}.

While the study of such states started purely out of fundamental curiosity \cite{fund_cur1,fund_cur2,fund_cur3}, they are now used in many applications in quantum information theory, such as quantum computing \cite{Gu2020}. They are also essential for early tests of quantum computing tasks \cite{quant_comp_tasks}, and quantum cryptography in quantum networks\cite{quant_networks}. Increasing the number of particles involved and the dimension of the GHZ state is essential both for foundational studies and practical applications. Motivated by this, a huge effort is being made by several experimental groups around the world to push the size of GHZ states. Photonic technology is one of the key technologies used to achieve this goal \cite{quant_comp_tasks,10photon}.

In 2017, Krenn, Gu and Zeilinger \cite{Quantum_graphs} discovered (and later extended \cite{Quantum_graphs_2, Quantum_graphs_3}) a bridge between experimental quantum optics and graph theory. They observed that large classes of quantum optics experiments (including those containing probabilistic photon pair sources, deterministic photon sources and linear optics elements) can be represented as edge-coloured graphs when there is \emph{no destructive interference}. Conversely, every edge-coloured graph can be translated into a concrete experimental setup. Suppose an experiment $\mathcal{E}$ corresponds to the edge-coloured graph $G_c$. $\mathcal{E}$ would lead to a valid GHZ state if and only if $c$ is PMValid on $G$. Moreover, if $c$ is a $k$-edge colouring, then the dimension of the GHZ state is $k$, and the number of particles in the GHZ state is equal to the number of vertices in $G$. Therefore the following question is equivalent to \cref{graph_question} 

\begin{question}
Can one characterize experiments (under the conditions described in \cite{Quantum_graphs}) which would lead to $|GHZ_{n,d}\rangle$, a $d$-dimensional $n$-particle GHZ state?
\end{question}

Vardi and Zhang \cite{vardi1, vardi2} also studied the computational complexity of several questions, which are inspired by this class of experiments. We note that while there are some experiments to construct GHZ states which does not fall into this category, these edge-coloured graphs capture many prominent ones. There is also an edge-weighted extension of this problem which arises from the discovery of new quantum interference effects \cite{Quantum_graphs_2}. Our characterisation does have some implications to this edge-weighted extension which we discuss in \cref{conclusion}. 

\begin{remark}
For readers who are more familiar with the language of quantum physics, this sub-section might be too short; so, they are referred to the appendix or the earlier works \cite{Quantum_graphs, Quantum_graphs_2, Quantum_graphs_3, krenn2019questions}. The authors of this paper deal with the graph theoretic question posed by quantum physicists as it is; we are not well versed in the quantum physical aspects of this question for discussing the extent of its applications, impact or significance. To the best of our knowledge, finding a counterexample to Krenn-Gu conjecture (\cref{KGconj}) has some serious implications to the construction of GHZ states and proving Krenn-Gu conjecture in affirmation has implications in quantum resource theory. 
\end{remark}

\subsection{Notation}
For a graph $G$, let $V(G), E(G)$ denote the set of vertices and edges, respectively.
For $S\subseteq V(G)$, $G[S]$ denotes the induced subgraph of $G$ on $S$. The cardinality of a set $\cal{S}$ is denoted by $|\cal{S}|$. For a positive integer $r$, $[r]$ denotes the set $\{1,2\ldots,r\}$. We refer to a cycle with $n$ vertices as $C_n$ and the complete graph with $n$ vertices as $K_n$.  $d(v)$ denotes degree of a vertex $v$. $\delta(G)$ denotes the minimum degree of $G$ and $\Delta(G)$ denotes the maximum degree of $G$.

If $c$ is PMvalid $k$-edge-colouring on $G$, we say that $G_c$ is \textit{perfectly monochromatic} and $\mu(G,c)=k$. The reader may see that $\mu(G)$ is the maximum value of $\mu(G,c)$ over all edge-colourings $c$, which are PMvalid.

A graph is matching covered if every edge of it is part of at least one perfect matching. If an edge $e$ is not part of any perfect matching $M$, then we call the edge $e$ to be redundant. By removing all redundant edges from the given graph $G$, we get its unique maximum matching covered sub-graph $mcg(G)$. It is easy to see that $\mu(G)=\mu(mcg(G))$.

\subsection{Our results}
It is easy to see that if $G$ has at least one perfect matching, a trivial $1$-edge-colouring on $G$ is PMValid. Hence $\mu(G)=0$ if and only if $G$ does not have a perfect matching. This implies that for every graph $G$, which is non-isomorphic to $K_4$ and has a perfect matching, $\mu(G)$ is either $1$ or $2$ (from \cref{bogdanov}). However, the characterization of graphs based on the value of their matching index being $1$ or $2$ remained unknown. Our main result is such a characterization in \cref{2_characterization_result} of \cref{sec:graph_struct}, which in turn gives a complete characterization of all graphs based on the value of their matching index.

We also give a fast algorithm which runs in time $O(|V(G)||E(G)|)$ to find the matching index of a graph in \cref{algorithmandproof}. The bottleneck of our runtime is the sub-routine of finding the matching covered graph of $G$. 

Our characterization has an interesting application to the more general version of the problem i.e., the weighted case discussed in \cref{conclusion}. Note that the weighted case corresponds to the situation when the so-called destructive interference is allowed (equivalently when the edge weights are allowed to be complex real numbers). Where as the unweighted case translates to when destructive interference is not allowed. Krenn and Gu conjectured that even when complex weights are allowed, the dimension can not be greater than $2$ (except for $K_4$). We show how to use our characterization to resolve Krenn-Gu conjecture for the class of graphs $\cal{H}$ (defined in \cref{sec:graph_struct}). That is if $G$ is a graph such that in the unweighted case the matching index is $2$, then even by allowing complex weights on the edges or by allowing a more general edge colouring (details in \cref{conclusion}), surprisingly the dimension can not become higher than $2$. This fact is stated formally in \cref{type_2_resolution} in \cref{conclusion}. The consequence of our charcaterization to the general version of the quantum physics problem was appreciated by quantum physicists when they awarded the the first Quantum-Graph best paper award for making progress towards Krenn-Gu conjecture \cite{krenn_website}.

\section{Structure of graphs with matching index $2$}
\label{sec:graph_struct}


\begin{figure}[H]
    \centering
    \begin{minipage}{0.45\textwidth}
        \centering
       \includegraphics[width=55mm,height=55mm]{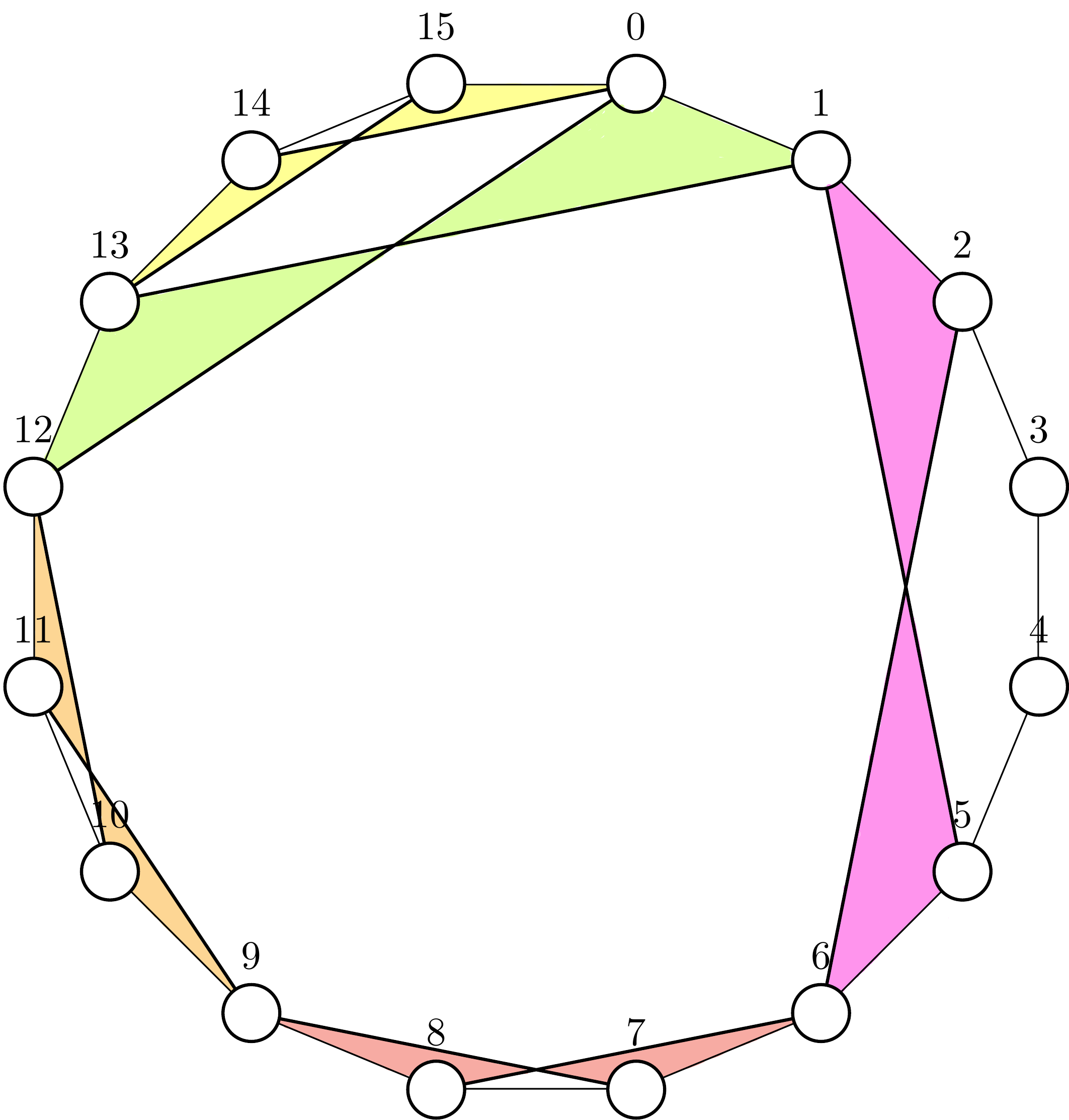}
    \end{minipage}
    \hspace*{1cm}
    \begin{minipage}{0.45\textwidth}
        \centering
\includegraphics[width=55mm, height=55mm]{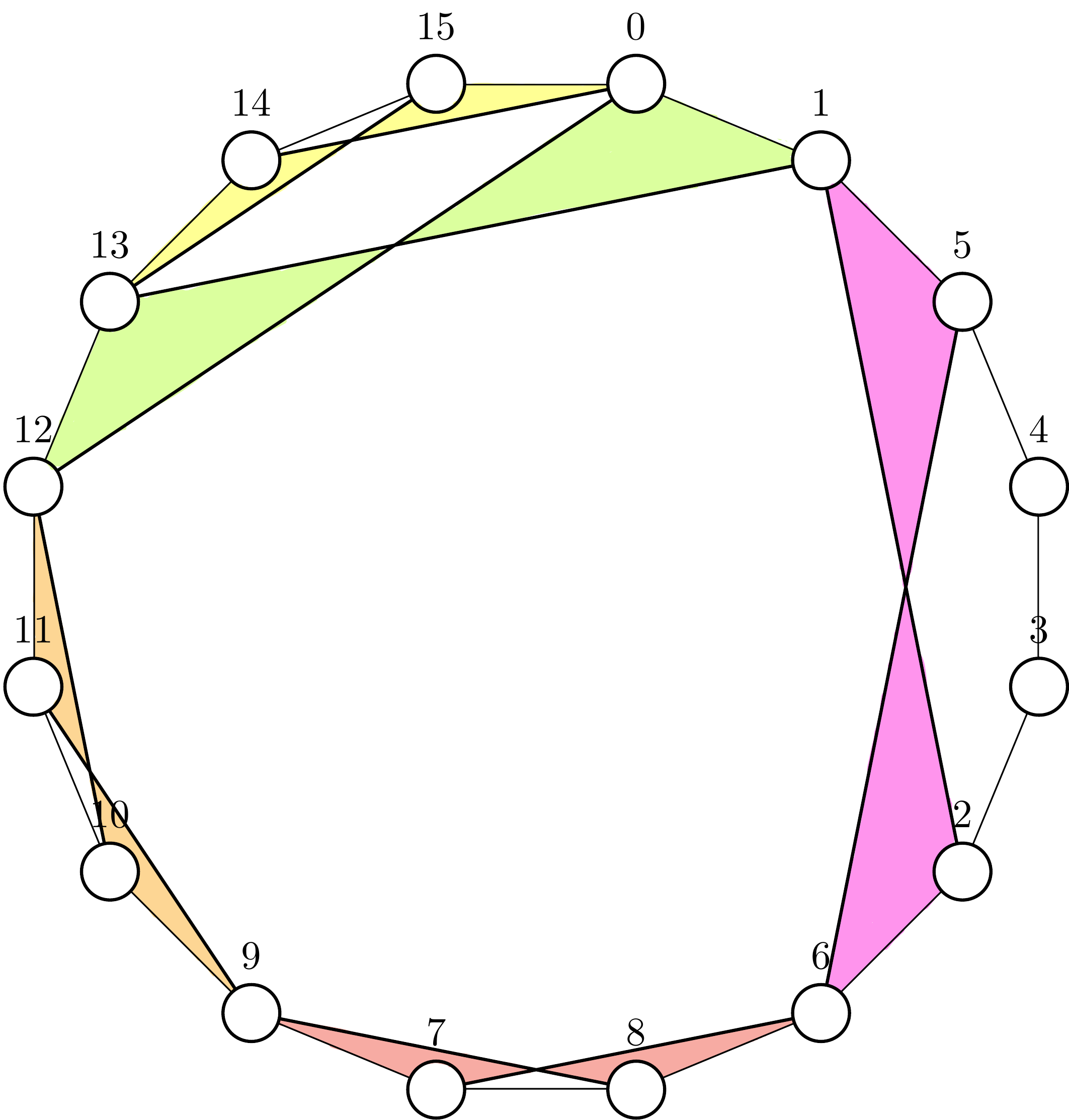}
\end{minipage}   
\caption{A graph $G$ satisfying
\cref{property1s}, \cref{property2s} and \cref{property3s} with respect to the cyclic orderings of two different Hamiltonian cycles. Note that in the figure on left, the ordering is with respect to the Hamiltonian cycle $C_1$ and the edge $e=\{1,5\}$ is a legal edge with respect to $G,C_1$. In the figure on right, the ordering is with respect to the Hamiltonian cycle $C_2$ and the same edge $e$ is now labelled $\{1,2\}$. Note that $e$ is a $C_2$-edge and not a legal edge with respect to $G,C_2$. Each drum is shaded with a different colour.}
\label{fig:type_h_graph}
\end{figure}

Let $G$ be a Hamiltonian graph of even order. Consider any Hamiltonian cycle $C$ and label its vertices be $0,1,2,3\ldots 2n-1$ in a clockwise direction along $C$. We interchangeably refer to vertex $i \mod 2n$ as $i$. The vertices $\{1,3\ldots,2n-1\}$ are referred to as odd vertices with respect to $C$ and the vertices $\{0,2\ldots,2n-2\}$ are referred to as even vertices with respect to $C$. The edges of the Hamiltonian cycle $C$ are defined to be \emph{$C$}-edges. The edges between two vertices of the same parity are defined to be \emph{legal} edges with respect to $C$. Every edge $e=\{i,j\}$ partitions the remaining vertices of graph into two parts ${\cal{P}}(e)=\{i+1,i+2,\ldots j-1\}$ and ${\cal{P}}'(e)=\{j+1,j+2\ldots,i-1\}$ with respect to $C$. If an edge $e'$ has one endpoint in ${\cal{P}}(e)$ and the other endpoint in ${\cal{P}}'(e)$, then $e'$ and $e$ are said to be \emph{crossing} each other with respect to $C$ and the pair of edges $e,e'$ is referred to as a \emph{crossing} pair with respect to $C$. 

Let the legal edges (with respect to $C$) $e=\{2i_1+1,2i_2+1\}$ and $e'=\{2j_1,2j_2\}$ cross each other (with respect to $C$) and without loss of generality, assume that the vertices $2i_1+1,2j_1,2i_2+1,2j_2$ appear in clockwise order on $C$. We say that $e,e'$ form a \emph{drum} with respect to $C$ if $e'=\{2j_1,2j_2\}=\{2i_2,2i_1\}$  or $e'=\{2j_1,2j_2\}=\{2i_1+2,2i_2+2\}$. In other words, an end vertex of $e$ and an end vertex of $e'$ must be adjacent on the cycle (that is, there should be a $C$-edge between them) and the remaining end vertices of $e,e'$ must also be adjacent on the cycle $C$. 

Let ${\cal{H}}$ be the family of even-order matching-covered Hamiltonian graphs with at least $6$ vertices such that all graphs $G \in {\cal{H}}$ satisfy the following: There \textit{exists} a Hamiltonian cycle $C$ of $G$ such that \cref{property1s}, \cref{property2s} and \cref{property3s} are satisfied with respect to the vertex labelling obtained from a clockwise cyclic ordering of $C$.

\begin{property}
\label{property1s}
All edges of $G$ are $C$-edges or legal edges with respect to $C$.
\end{property}

\begin{property}
\label{property2s}
Every crossing pair (with respect to $C$) of $G$ forms a drum with respect to $C$.
\end{property}

\begin{property}
\label{property3s}
Each legal edge (with respect to $C$) of $G$ is part of exactly one drum with respect to $C$.
\end{property}

An example graph $G$ is shown in \cref{fig:type_h_graph}, which satisfies \cref{property1s}, \cref{property2s} and \cref{property3s} with respect to the cyclic orderings along two different Hamiltonian cycles. As the properties are satisfied with respect to at least one cycle, $G \in {\cal{H}}$. 
This special class of graphs also happens to be equivalent to the set of matching covered graphs with a matching index of $2$.

\begin{theorem}
\label{2_characterization_result}
For any graph $H$ with $V(H)>4$, $\mu(H)=2$ if and only if $mcg(H)\in {\cal{H}}$.
\end{theorem}
We prove \cref{2_characterization_result} in \cref{sec:proof_of_algorithm}. 

\section{Proof of \cref{2_characterization_result}}
\label{sec:proof_of_algorithm}

\subsection{Preliminaries}
\label{subsec:proof_of_prelim}
Throughout this section, we assume that the graph $G$ is not isomorphic to $K_4$. We first prove some simple structural properties of an edge-coloured graph $G_c$ based on its $\mu(G,c)$ value. 

If an even cycle contains disjoint monochromatic perfect matchings of different colours, we call it a cycle with alternating colours. We now present a simple theorem about edge-coloured graphs.
\begin{theorem}
\label{simple_classification_theorem}
For an unweighted edge-coloured graph $G_c$, where $G$ is the underlying graph non-isomorphic to $K_4$ and $c$ is the colouring, the statements $1$ through $4$ are true about the following conditions.
\setcounter{condition}{0}
\begin{condition}
\label{pm_graph}
$G_c$ is perfectly monochromatic.
\end{condition}

\begin{condition}
\label{pm_cond}
$G_c$ has a perfect matching.
\end{condition}

\begin{condition}
\label{ham_cond}
$G_c$ has a Hamiltonian cycle with alternating colours. 
\end{condition}
\begin{enumerate}
    \item $\mu(G,c)$ is not defined if and only if \cref{pm_graph} is false.
    \item $\mu(G,c)=0$ if and only if \cref{pm_graph} is true and \cref{pm_cond} is false.
    \item $\mu(G,c)=1$ if and only if \cref{pm_graph}, \cref{pm_cond} are true and \cref{ham_cond} is false. 
    \item $\mu(G,c)=2$ if and only if \cref{pm_graph}, \cref{pm_cond} and \cref{ham_cond} are true.
\end{enumerate}

\end{theorem}
\begin{proof}
Statement $1$ is true by definition. 

If \cref{pm_graph} is true and \cref{pm_cond} is false, as there are no perfect matchings, $\mu(G,c)=0$. If $\mu(G,c)=0$, $G_c$ must be perfectly monochromatic by Statement $1$. Therefore, if there is a perfect matching it must be monochromatic. But that would imply $\mu(G,c) \geq 1$. Therefore, there is no perfect matching in $G_c$.  Therefore, statement $2$ is true.

We prove that if \cref{pm_graph}, \cref{pm_cond} are true and \cref{ham_cond} is false, then $\mu(G,c)=1$. If \cref{pm_graph} and \cref{pm_cond} are true then $\mu(G,c) \geq 1$ by statement $1$ and statement $2$. Towards a contradiction, let  $\mu(G,c) \neq 1$. It follows that $\mu(G,c) \geq 2$. Therefore, there are at least two colours, say red and blue, with a monochromatic perfect matching. Consider the union of a red perfect matching and a blue perfect matching. It must be a disjoint union of even cycle(s) with alternating colours. Clearly, each cycle can be decomposed into two perfect matchings of different colours. If it is a disjoint union of at least two cycles, then by selecting the red perfect matching from one cycle and blue perfect matchings from the remaining cycles, we can construct a non-monochromatic perfect matching. This contradicts the assumption that \cref{pm_graph} is true. Therefore, the union of a red and a blue perfect matching is a Hamiltonian cycle with edges of alternating colours. But it is given \cref{ham_cond} is false. Contradiction.

If \cref{pm_graph}, \cref{pm_cond} and \cref{ham_cond} are true, since there are two monochromatic perfect matchings of different colours which can be formed by selecting the edges of the same colour from the Hamiltonian cycle with alternating colours, $\mu(G_c) \geq 2$. As \cref{pm_graph} is true, we know that $\mu(G,c) \leq 2$ by \cref{bogdanov}. Therefore, $\mu(G,c) = 2$. 

If $\mu(G,c)=1$, then it is easy to see that \cref{pm_graph}, \cref{pm_cond} are true. \cref{ham_cond} must be false, because if it is true, that would imply $\mu(G,c)=2$ as proved above. Similarly, if $\mu(G,c)=2$, then it is easy to see that \cref{pm_graph}, \cref{pm_cond} are true. \cref{ham_cond} must be true, because if it is false, that would imply $\mu(G,c)=1$ as proved above.
\end{proof}

\subsection{Definitions}
\label{type_2_defintions}


For a graph $G$ of even order and with a Hamiltonian cycle $C$, we define some more notation with respect to $G$ and $C$. Let the vertices be $0,1,2,3\ldots 2n-1$ in clockwise direction along the cycle $C$. We interchangeably refer vertex $i \mod 2n$ as $i$. An odd-even edge is called \emph{illegal} if it is not a $C$-edge.  Recall that the odd-odd edges and even-even edges are called \emph{legal} edges. It is easy to see that the edges of the graph can be classified into legal edges, illegal edges and $C$-edges.

Recall that every edge $e=\{i,j\}$ partitions the remaining vertices of cycle into two parts ${\cal{P}}(e)=\{i+1,i+2,\ldots j-1\}$ and ${\cal{P}}'(e)=\{j+1,j+2\ldots,i-1\}$. The partition weight of $e$ is defined to be equal to $\min \{|{\cal{P}}(e)|,|{\cal{P}}'(e)|\}$. Note that both $|{\cal{P}}(e)|,|{\cal{P}}'(e)|$ are odd (even) for legal (illegal) edges. If $|i-j|$ is odd, we define ${\mathbf{M}}_{i,j}=\{\{i,i+1\},\{i+2,i+3\}\ldots,\{j-1,j\}\}$, a perfect matching formed by the alternate $C$-edges in clockwise order from $i$ to $j$.

Recall that if an edge $e'$ has one endpoint in ${\cal{P}}(e)$ and the other endpoint in ${\cal{P}}'(e)$, then $e'$ and $e$ are said to be crossing each other. We say that a pair of crossing edges form a \emph{legal crossing pair} if both the edges are legal. We say that a \emph{legal} crossing pair is a \emph{nice} crossing pair if one of the edges is odd-odd and the other edge is even-even.


Recall that a drum $D$, say with vertices $2i_1+1,2i_1+2,2i_2+1,2i_2+2$ partitions the remaining vertices of the cycle into two parts, ${\cal{P}}(D)=\{2i_2+3,2i_1+4,\ldots 2i_1\}$ and ${\cal{P}}'(D)=\{2i_1+3,2i_2+4\ldots 2i_2\}$. It follows that $|{\cal{P}}(D)|,|{\cal{P}}'(D)|$ are even for a drum by the definition. We refer to the paths along $C$ from $2i_2+2$ to $2i_1+1$ through ${\cal{P}}(D)$ and from $2i_1+2$ to $2i_2+1$ through ${\cal{P}}'(D)$ to be straps of the drum. The length of a strap is the number of edges on the path. We refer to $\min \{|{\cal{P}}(D)|,|{\cal{P}}'(D)|\}$ as the partition weight of the drum $D$. Note that a drum is a $4$-cycle with $2$ $C$-edges forming a pair of opposite edges and a nice crossing pair forming the remaining two opposite edges.

It is important to note that even though there might be structures similar to drums formed by two illegal crossing edges, which partition the cycle into two odd parts, we \emph{do not} refer to them as drums. Only nice crossing pairs can form drums. It follows that, if we colour the edges of the Hamiltonian cycle with alternating colours, both the $C$-edges of a drum will get the same colour. We make a few observations regarding the Hamiltonian cycles of $G$ in \cref{sec:observations}

\vspace{3mm}

\fbox{
\parbox{0.9\textwidth}{
We emphasize that legal edges, illegal edges, $C$-edges, crossing pairs and drums are defined with respect to a graph and a vertex labelling along the clockwise direction of a Hamiltonian cycle $C$. For example, for a graph $G$, the legal edges with respect to $C$ might not be legal edges with respect to a different Hamiltonian cycle $C'$. See \cref{fig:type_h_graph}.}}


\subsection{Observations}
\label{sec:observations}

We refer to a Hamiltonian cycle with alternating red-blue edges as a red-blue alternating Hamiltonian cycle $C$. For graphs with such a Hamiltonian cycle $C$, whose vertices are labelled in clockwise cyclic order along $C$, we use a colouring convention that the edges of $C$ of the form $\{2i-1,2i\}$ be red and the edges $\{2i,2i+1\}$ be blue. We call a drum in $G_c$ monochromatic if all $4$ edges of the drum are all of the same colour.

\begin{condition}
\label{condition_for_observations}
$G$ is a matching covered graph with a colouring $c$ for which $\mu(G,c)=2$ and {$|V(G)|=2n\geq 6$}. From \cref{simple_classification_theorem}, it has at least one red-blue alternating Hamiltonian cycle, say $C$. Label the vertices in clockwise cyclic order with respect to $C$ so that colouring convention is satisfied.
\end{condition}
We now make some observations about unweighted edge-coloured graphs that satisfy \cref{condition_for_observations}, which we will later use to prove the structural characterization of graphs where $\mu(G)=2$. 
\begin{remark}
Note that all the observations made below for $G_c$ are with respect to a vertex labelling along the clockwise cyclic ordering of red-blue alternating Hamiltonian cycle $C$ from \cref{condition_for_observations}. For instance, in \cref{odd_even_edges}, we only claim that there are no illegal edges with respect to a cyclic ordering along $C$. 
\end{remark}

\begin{observation}
\label{odd_even_edges}
If \cref{condition_for_observations} is true, $G_c$ has no illegal edges.
\end{observation}
\begin{proof}
Towards a contradiction, let there be an illegal edge. Let it be $\{1,2j\}$ for some $j\in[2,n-1]$, without loss of generality. Consider the perfect matching $M=\{\{1,2j\}\}\bigcup {\mathbf{M}}_{[2,2j-1]} \bigcup {\mathbf{M}}_{[2j+1,0]}$, see \cref{type_2_defintions} for notation. As an illegal edge cannot be a $C$-edge, both ${\mathbf{M}}_{[2,2j-1]}, {\mathbf{M}}_{[2j+1,2n]}$ are non-empty. Observe that ${\mathbf{M}}_{[2,2j-1]}$ has colour blue and ${\mathbf{M}}_{[2j+1,0]}$ has colour red by the colouring convention of $C$. Therefore, the perfect matching $M$ is a non-monochromatic perfect matching. But all perfect matchings of $G_c$ are monochromatic as $G_c$ is perfectly monochromatic. Contradiction.
\end{proof}

\begin{observation}
\label{nice_pairs_are_drums}
If \cref{condition_for_observations} is true, all nice crossing pairs of $G_c$ must form monochromatic drums.
\end{observation}
\begin{proof}
First we prove that all nice crossing pairs form a drum. Towards a contradiction, let there be a nice crossing pair $\{2i_1+1,2i_2+1\}$ and $\{2j_1,2j_2\}$ which do not form a drum. Without loss of generality, we assume that $2i_1+1,2j_1,2i_2+1,2j_2$ are in a clockwise direction on the cycle $C$. There is a perfect matching 
\begin{multline*}
    {\mathbf{M}}=\{\{2i_1+1,2i_2+1\},\{2j_1,2j_2\}\} \bigcup {\mathbf{M}}_{2i_1+2,2j_1-1}\bigcup {\mathbf{M}}_{2j_1+1,2i_2}\\ \bigcup {\mathbf{M}}_{2i_2+2,2j_2-1}\bigcup {\mathbf{M}}_{2j_2+1,2i_1}
\end{multline*}
We see that ${\mathbf{M}}_{2j_2+1,2i_1}\bigcup {\mathbf{M}}_{2j_1+1,2i_2}$ is contained in the red colour class and ${\mathbf{M}}_{2i_1+2,2j_1-1}\bigcup{\mathbf{M}}_{2i_2+2,2j_2-1}$ is contained in blue colour class by the colouring convention of $C$. If either of these two is empty, the crossing pair forms a drum. Therefore both are non-empty, and $M$ is a non-monochromatic perfect matching. But all perfect matchings of $G_c$ are monochromatic as $G_c$ is perfectly monochromatic. This is a contradiction. Therefore, all nice crossing pairs must form drums.

Now we show that every drum is monochromatic. Towards a contradiction, let there be a non-monochromatic drum. Recall (see the definition of a drum) that non-adjacent $C$-edges $e_1, e_2$ of a drum are of the same colour, say red. Therefore, $e_1, e_2$ are part of the red perfect matching $M$ formed by all red edges of $C$. Let the crossing pair of the drum be $e_3,e_4$. Consider the perfect matching $M'=M-\{e_1,e_2\}\bigcup\{e_3,e_4\}$. Since $G$ has at least $6$ vertices, $M-\{e_1,e_2\}$ has at least one edge of the red colour. As it is a non-monochromatic drum, at least one of $e_3,e_4$ must be blue. Therefore, $M'$ is a non-monochromatic perfect matching. But all perfect matchings of $G_c$ are monochromatic as $G_c$ is perfectly monochromatic. This is a contradiction. Therefore, all nice crossing pairs must form monochromatic drums.
\end{proof}

\begin{observation}
\label{exactly_one_drum}
If \cref{condition_for_observations} is true, every legal edge of $G_c$ is contained in exactly one drum.
\end{observation}
\begin{proof}
Let the legal edge be $e=\{1,2j+1\}$ and of colour red, without loss of generality. Recall that $e$ can only be part of the drums with non-adjacent cycle edges $\{1,2\},\{2j+1,2j+2\}$ or $\{0,1\},\{2j,2j+1\}$ by definition. We claim that the latter possibility can be ruled out. This is because $e$ cannot form a drum with the legal edge $e'=\{0,2j\}$, since $\{e,e'\}$ would form a nice crossing pair and the drum formed by $\{e,e'\}$ should be monochromatic by \cref{nice_pairs_are_drums}; However, $\{0,1\},\{2j,2j+1\}$ which would form the $C$-edges of that drum are of colour blue by the colouring convention of $C$. Therefore, $e$ can be part of at most one drum. A similar argument holds even if $e$ was of the colour blue or $e$ was an even-even edge.

We define a legal edge $e_1$ to be saturated if, for every perfect matching ${M}$ containing $e_1$, there exists another legal edge $e_2 \in {M}$ forming a drum with $e_1$. We now claim that all legal edges are saturated. Suppose not. Without loss of generality, let $e=(1,2j+1)$ have the minimum partition weight among all the legal edges that are not saturated. Let $[2,2j]$ be the smaller among ${\cal{P}}(e)$ and ${\cal{P}}'(e)$, without loss of generality. Let ${M}$ be the perfect matching which makes $e$ not saturated, i.e. there is no edge $e'$ in ${M}$, which forms a drum with $e$. By assumption on minimality of $e$, all the edges part of ${M}$ with both its vertices in $[2,2j]$ would be $C$-edges and the pairs of non-$C$-edges of drums. It is easy to see that a pair of non-$C$-edges of drums match the same number of odd vertices and even vertices, and a $C$-edge also matches the same number of odd vertices and even vertices. Therefore, all the edges part of ${M}$ with both its vertices in $[2,2j]$ match the same number of odd vertices and even vertices in  $[2,2j]$. However, as there are $j$ even vertices and $j-1$ odd vertices in $[2,2j]$, an even vertex in $[2,2j]$ must match with a vertex in $[2j+2,0]$. Such a vertex in $[2j+2,0]$ must also be even since illegal edges are absent by \cref{odd_even_edges}. It follows that this even-even edge crosses $e$, and since $e$ is an odd-odd edge, they form a nice crossing pair. Therefore, they must form a drum from \cref{nice_pairs_are_drums} and hence $e$ is saturated. Contradiction. 

Since the graph is matching covered, we know that every legal edge is part of at least one perfect matching. As every legal edge is saturated, it follows that every legal edge is also part of a drum.
\end{proof}

\begin{observation}
\label{cycle_edge_in_one_drum}
If \cref{condition_for_observations} is true, each $C$-edge of $G_c$ can be part of at most one drum.
\end{observation}
\begin{proof}
Let a $C$-edge $e$ be part of two drums, say $D_1,D_2$ towards a contradiction. From \cref{exactly_one_drum}, the legal edges in both drums must be disjoint. Let they be $e_1,e_1'$ in drum $D_1$ and $e_2,e_2'$ in drum $D_2$. Let $e_1,e_2$ be the edges incident on even end point of $e$ and $e_1',e_2'$  be the edges incident on the odd end point of $e$. Observe that one of $\{e_1,e_2'\}$ and $\{e_1',e_2\}$ cross each other. It follows that they form a nice crossing pair and hence form a drum $D_3\notin \{D_1,D_2\}$ by \cref{nice_pairs_are_drums}. However, a legal edge can be part of only one drum by \cref{exactly_one_drum}. Contradiction.
\end{proof}

\begin{observation}
\label{legal_are_nice}
If \cref{condition_for_observations} is true, all legal crossing pairs of $G_c$ are nice.
\end{observation}
\begin{proof}
Towards a contradiction, let there be a legal crossing pair $e,e'$, which is not nice. Without loss of generality, let both of them be odd-odd edges. From \cref{exactly_one_drum}, $e$ must form a drum $D$ with a even-even edge, say $e''$. As $e,e'$ cross, they can not be incident on the same vertex. As $e'$ is odd-odd and $e''$ is even-even, they can not be incident on the same vertex. Therefore, $e'$ should be incident on a vertex from ${\cal{P}}(D)$ and a vertex from ${\cal{P}}'(D)$. Therefore, $e',e''$ cross each other. As they are a nice crossing pair, from \cref{nice_pairs_are_drums}, they must form a drum. Therefore, $e''$ is part of two drums. But from \cref{exactly_one_drum}, $e''$ is part of exactly one drum. Contradiction.
\end{proof}
\begin{observation}
\label{crossing_are_mono}
If \cref{condition_for_observations} is true, all crossing pairs of $G_c$ form monochromatic drums.
\end{observation}
\begin{proof}
All crossing pairs are legal since there are no illegal edges by \cref{odd_even_edges}. All legal crossing pairs are nice from \cref{legal_are_nice}. All nice crossing pairs form monochromatic drums from \cref{nice_pairs_are_drums}. Therefore, all crossing pairs form monochromatic drums.
\end{proof}

\begin{observation}
\label{red_2_blue_2}
If \cref{condition_for_observations} is true, every vertex $v$ of $G_c$ can have at most two red edges and at most two blue edges incident on it. Moreover, if the vertex has two red(blue) edges incident on it, both of those red(blue) edges must be part of a red(blue) drum.
\end{observation}
\begin{proof}
Let the two cycle edges incident on $v$ be $e_1,e_2$. Without loss of generality, let $e_1$ be red and $e_2$ be blue. Every $C$-edge can be part of at most one drum from \cref{cycle_edge_in_one_drum}. Also, all drums are monochromatic from \cref{crossing_are_mono}. Therefore, $e_1$ can be part of at most one red drum and $e_2$ can be part of at most one blue drum.

Recall that other than $C$-edges, only legal edges are incident on $v$ from \cref{odd_even_edges}. Every legal edge incident on vertex $v$ must be part of a unique drum from \cref{exactly_one_drum}. But each of these drums must contain a cycle edge incident on $v$. It follows that $v$ can be part of at most one red drum and at most one blue drum, and therefore, $v$ can have at most two red edges and two blue edges incident on it.
\end{proof}

\begin{observation}
\label{max_degree_4}
If $G$ is a matching covered graph and $\mu(G)=2$, then $\Delta(G)\leq 4$.
\end{observation}
\begin{proof}
If $|V(G)|<6$, it is easy to see that $\Delta(G)\leq 4$. We now assume that $|V(G)|\geq 6$. As $\mu(G)=2$, there is a colouring $c$ such that $\mu(G,c)=2$. It follows from \cref{simple_classification_theorem} that \cref{condition_for_observations} holds. Therefore, from \cref{red_2_blue_2}, on each vertex of $G_c$, at most two red edges and at most two blue edges are incident. Therefore, the degree is $4$ for any vertex on $G$ and hence $\Delta(G)\leq 4$.
\end{proof}

\begin{observation}
If \cref{condition_for_observations} is true, all Hamiltonian cycles of $G_c$ are red-blue alternating cycles.
\label{all_cycles_are_alternating}
\end{observation}
\begin{proof}
Any Hamiltonian cycle can be partitioned into two perfect matchings. Since $G_c$ is perfectly monochromatic, both these perfect matchings must be monochromatic. Therefore, the Hamiltonian cycle is either a red-blue alternating cycle or a monochromatic cycle.

There is a red-blue alternating Hamiltonian cycle $C_2$ from \cref{simple_classification_theorem}. Towards a contradiction, let there be a monochromatic cycle $C_1$, say of colour red, without loss of generality. Therefore, every vertex must have at least two red edges incident on it. From \cref{red_2_blue_2}, every vertex can have at most two red edges incident on it. It follows that all red edges of the graph form a Hamiltonian cycle $C_1$. From \cref{red_2_blue_2}, since all vertices have $2$ red edges incident on them, each vertex must be part of a drum with respect to $C_2$. Therefore, the monochromatic red Hamiltonian cycle $C_1$ must contain a red monochromatic drum, which is a $4$-cycle. This is possible only when $C_1$ is a $4$-cycle. But $|V(G)|\geq 6$ from \cref{condition_for_observations}. Contradiction. It follows that there are no monochromatic red Hamiltonian cycles. Using a similar argument, it can be proved that there are no monochromatic blue Hamiltonian cycles. Therefore, all Hamiltonian cycles of $G_c$ are red-blue alternating cycles. 
\end{proof}

\begin{observation}
\label{4_cycle_observation}
If \cref{condition_for_observations} is true and $G$ is not a cycle, a drum $D$ exists such that all vertices in ${\cal{P}}(D)$ have degree $2$.
\end{observation}
\begin{proof}
Among all drums, let $D$ be a drum having smallest partition weight equal to $|{\cal{P}}(D)|$. Let $D$ be formed by the crossing edges $e_1,e_2$. If ${\cal{P}}(D)=\emptyset$, the observation trivially holds. Therefore, we assume that ${\cal{P}}(D)$ has at least $1$ vertex. Towards a contradiction, assume that there is a vertex $v$ in ${\cal{P}}(D)$ with degree greater than $2$. Therefore, there exists some legal edge $e_3$, incident on $v\in {\cal{P}}(D)$.

From \cref{exactly_one_drum}, $e_3$ must be part of some drum $D'$ along with a crossing edge $e_3'$. As $e_1,e_2$ are part of exactly one drum $D$, $e_3'$ or $e_3$ cannot be in $\{e_1,e_2\}$, as in that case there would another drum involving $e_1$ or $e_2$ other than $D$. From \cref{crossing_are_mono}, as all crossing pairs are drums, neither $e_3$ nor $e_3'$ can cross $e_1$ or $e_2$. 
Therefore, both the $C-$edges of $D'$ must be on the strap of $D$ along ${\cal{P}(D)}$. Therefore, $|{\cal{P}}(D)| > |{\cal{P}}(D')|$. But this contradicts the assumption that $D$ has the smallest partition weight.
\end{proof}

\subsection{Structural characterization of graphs for which $\mu(G)=2$}
\label{type_2_structure}
Recall the definitions of $\cal{H}$ from \cref{sec:graph_struct}. Let ${\cal{H}}'$ be the family of even-order matching-covered Hamiltonian graphs with at least $6$ vertices such that all graphs $G \in {\cal{H}}'$ satisfy the following:
For \textit{all} Hamiltonian cycles $C$ of $G$, \cref{property1s}, \cref{property2s} and \cref{property3s} are satisfied with respect to the vertex labelling obtained from a clockwise cyclic ordering of $C$. 

Observe that the \cref{property1s}, \cref{property2s} and \cref{property3s} are satisfied for all cycles for ${\cal{H}}'$ and some cycle for ${\cal{H}}$. It is easy to see that $ {\cal{H}}' \subseteq {\cal{H}}$. Surprisingly, we show that the two classes ${\cal{H}}'$ and ${\cal{H}}$ are, in fact, the same!
\begin{lemma}
\label{structure_proof_dir_1}
For any graph $H$ with $|V(H)|>4$, if $\mu(H)=2$, then $mcg(H) \in {\cal{H}'}$
\end{lemma}
\begin{proof}
    Let $G=mcg(H)$. If $\mu(H)=2$, then $\mu(G)=2$. It follows that there exists a perfectly monochromatic colouring $c$ such that $\mu(G,c)=2$.
From \cref{simple_classification_theorem}, $G$ has a perfect matching implying that the order of $G$ is even; Moreover, $G_c$ has at least one alternating red-blue Hamiltonian cycle. Therefore \cref{condition_for_observations} holds for $G_c$. From \cref{all_cycles_are_alternating}, all Hamiltonian cycles of $G_c$ are alternating red-blue cycles. Therefore, with respect to any Hamiltonian cycle of $G$ \cref{condition_for_observations} holds. It follows that \cref{property1s}, \cref{property2s} and \cref{property3s} hold due to \cref{odd_even_edges}, \cref{crossing_are_mono} and \cref{exactly_one_drum} respectively. Therefore, $mcg(H) \in {\cal{H}'}$.
\end{proof}

\begin{lemma}
\label{structure_proof_dir_2}
If $mcg(H) \in {\cal{H}}$, then $\mu(H)=2$.
\end{lemma}
\begin{proof}
Let $G=mcg(H)$. We first show that for any graph $G\in {\cal{H}}$, there is an edge colouring $c$ such that $\mu(G,c)$ is $2$.

By definition of ${\cal{H}}$, there exists an even order Hamiltonian cycle, $C$ with respect to which \cref{property1s}, \cref{property2s} and \cref{property3s} are satisfied. Colour $C$ with alternating red and blue colours. Notice that all uncoloured edges are legal due to \cref{property1s}. Consider a legal edge $e$, due to \cref{property3s}, it is part of a unique drum $D$ having another legal edge $e'$, where $e,e'$ form a crossing pair. By the definition of a drum, both the non-adjacent $C$-edges of $D$ must be of the same colour with respect to colouring $c$. Colour $e,e'$ with the same colour as the non-adjacent $C$-edges of $D$. Do this for all legal edges. Due to \cref{property3s}, as each legal edge is part of exactly one drum, each legal edge gets coloured exactly once. Thus, we now have a colouring $c$ for the entire graph. We now prove that this colouring $c$ is perfectly monochromatic.

We first prove that, for a drum $D$ with a crossing pair $e,e'$ and a perfect matching $M$, either $\{e,e'\}\subseteq M$ or $\{e,e'\} \bigcap M=\emptyset$. Towards a contradiction, let there be a perfect matching $M$, which contains $e$ but not $e'$. As $e$ is legal, $|{\cal{P}}(e)|$ and $|{\cal{P}}'(e)|$ are odd. Since $M$ can not match the vertices of ${\cal{P}}(e)$ within itself and every vertex must be covered by $M$, there must be an edge $e''\in M$ crossing $e$. But every crossing pair forms a drum due to \cref{property2s} and $e$ is part of only one drum due to \cref{property3s}. Therefore $e''$ must be $e'$. Thus $e' \in M$. Contradiction.

We now prove that the colouring $c$ is perfectly monochromatic. Towards a contradiction, let there be a non-monochromatic perfect matching $M$. If $M$ contains legal edges, they come in pairs such that each pair belongs to a drum. Clearly, all these pairs of legal edges that belong to $M$ are disjoint from \cref{property3s}. Replace each such crossing pair in $M$ with the associated pair of $C$-edges of the corresponding drum to form a new perfect matching $M'$. Observe that since all drums are monochromatic, we are replacing a pair of crossing edges with $C$-edges of the same colour. It follows that since $M$ is a non-monochromatic perfect matching, $M'$ must also be a non-monochromatic perfect matching. However, $M'$ is a perfect matching from the Hamiltonian cycle $C$, which is alternately coloured. Thus $M'$ must be monochromatic. This is a contradiction.

Therefore, $G_c$ has no non-monochromatic perfect matchings. Therefore, $G_c$ is perfectly monochromatic. Since there are two perfect matchings of different colours in $C$, $\mu(G,c)=2$. Therefore, $\mu(G)\geq 2$. From \cref{bogdanov}, $\mu(G)\leq 2$. It follows that $\mu(G) = 2$ and hence $\mu(H) = 2$.
\end{proof}
\begin{lemma}\label{hish}
$\cal{H}=\cal{H}'$
\end{lemma}
\begin{proof}
It is easy to see that ${\cal{H}'} \subseteq {\cal{H}}$ by definition. Consider any matching covered graph $G \in {\cal{H}}$. From \cref{structure_proof_dir_2}, $\mu(G)=2$. As $\mu(G)=2$, from \cref{structure_proof_dir_1}, $G \in {\cal{H}'}$. It follows that ${\cal{H}} \subseteq {\cal{H}'}$ and hence ${\cal{H}}={\cal{H'}}$.
\end{proof}

By \cref{structure_proof_dir_1}, if $\mu(H)=2$, then $mcg(H) \in {\cal{H}'}$. But from \cref{hish}, ${\cal{H}}'={\cal{H}}$. Finally, from \cref{structure_proof_dir_2}, if $mcg(H) \in {\cal{H}}$ then $\mu(H)=2$. Therefore, \cref{2_characterization_result} follows.

\section{Algorithm to find the matching index}\label{algorithmandproof}

We first prove \cref{algo_proof_lemma}, which is useful to prove the correctness of our algorithm.
\begin{lemma}
\label{algo_proof_lemma}
Let $G$ be a matching-covered graph with ${\mu}(G)=2$. Let $v \in V(G)$ and ${\cal{M}}(e)$ denote the set of all perfect matchings of $G$ containing the edge $e\in E(G)$. There exists two edges $e_1,e_2$ incident on $v$ such that $M_1\bigsqcup M_2$ is a Hamiltonian cycle of $G$, for any $M_1 \in {\cal{M}}(e_1)$ and any $M_2 \in {\cal{M}}(e_2)$
\end{lemma}
\begin{proof}
Let $c$ be a perfectly monochromatic colouring of $G$ such that $\mu(G,c)=2$. From \cref{simple_classification_theorem}, $G_c$ is perfectly monochromatic and has a red-blue alternating Hamiltonian cycle $C$. For a vertex $v$, let two of its incident edges from $C$ be  $e_1,e_2$. Without loss of generality, let $e_1$ be red and $e_2$ be blue. For those $e_1,e_2$, we prove that $M_1\bigsqcup M_2$ is a Hamiltonian cycle of $G_c$, for any $M_1 \in {\cal{M}}(e_1)$ and $M_2 \in {\cal{M}}(e_2)$.

Since $e_1$ is red, and $G_c$ is perfectly monochromatic, all perfect matchings in ${\cal{M}}(e_1)$ must be monochromatic red perfect matchings by definition. Similarly, all perfect matchings in ${\cal{M}}(e_2)$ must be monochromatic blue perfect matchings. Therefore, $M_1\bigcap M_2=\emptyset$ and $M_1\bigsqcup M_2$ must be a disjoint union of even cycle(s) of alternating colours. If it is a disjoint union of at least two even cycles, then by selecting the red perfect matching from one cycle and blue perfect matchings from the remaining cycles, we can construct a non-monochromatic perfect matching. However, this is not possible as $G_c$ is perfectly monochromatic. Therefore, $M_1\bigsqcup M_2$ is a Hamiltonian cycle.
\end{proof}

\RestyleAlgo{ruled}
\SetKwComment{Comment}{/* }{ */}
\begin{algorithm}[H]
\label{appendix:algo}
\caption{To decide whether a non-trivial matching covered graph non-isomorphic to $K_4$ has matching index  $1$ or $2$}\label{alg:two}
Pick any $v \in V(G)$, arbitrarily.

${\cal{C}} \gets \emptyset$ 

\eIf{$d(v) \geq 5$}{$\mu(G) \gets 1$  \hfill\tcp{See \cref{max_degree_4}}
}
  {
  \For{ all pairs of edges $\{e_1,e_2\}$ incident on $v$}{
Pick arbitrary perfect matchings $M_1,M_2$ such that $e_1 \in M_1$ and $e_2 \in M_2$

  \If{$M_1\bigcup M_2$ is a Hamiltonian cycle}
   {${\cal{C}} \gets {\cal{C}}\bigcup \{M_1\bigcup M_2\}$}}
  }
  \eIf{${\cal{C}} = \emptyset$}{$\mu(G) \gets 1$ \hfill\tcp{See \cref{algo_proof_lemma}}}{
  Pick any $C \in {\cal{C}}$ arbitrarily \;
  \eIf{\cref{property1s}, \cref{property2s}, \cref{property3s} are satisfied in $G$ with respect to $C$}{$\mu(G) \gets 2$ \hfill\tcp{$G \in {\cal{H}'}$, See \cref{hish}  }}{$\mu(G) \gets 1$ \hfill\tcp{$G \notin {\cal{H}}$, See \cref{2_characterization_result} }}}
\end{algorithm}

\subsection{Proof of correctness}


\begin{proof}
From \cref{bogdanov}, if $H$ isomorphic to $K_4$, $\mu(H)=3$. From \cref{simple_classification_theorem}, $\mu(H)=0$ if and only if there is no perfect matching in $H$. The absence of a perfect matching can be checked using an algorithm of Micali and Vazirani \cite{DBLP:conf/focs/MicaliV80} in $O(\sqrt{|V|}|E|)$ time. 


We now know that $\mu(H)=1$ or $2$. We first find the maximum matching covered subgraph $G=mcg(H)$ in $O(|V||E|)$ time using a deterministic algorithm due to Carvalho and Cheriyan \cite{Carvalho2005AnOA}. It can also be computed in $O(|V|^{2.376})$ time using a randomized algorithm due to Rabin and Vazirani \cite{DBLP:journals/jal/RabinV89}. As, $\mu(G)=\mu(H)$, it is now sufficient to find $\mu(G)$ as described in \cref{alg:two}. The overview of \cref{alg:two} is described below, along with its correctness.


If $G$ is not Hamiltonian, $\mu(G)=1$ from \cref{simple_classification_theorem}. If $G$ is Hamiltonian and we know a Hamiltonian cycle $C$ of $G$, it is easy to check if $G$ satisfies \cref{property1s}, \cref{property2s} and \cref{property3s} with respect to $C$ in $O(|E|)$ time. Due to \cref{2_characterization_result}, if all three properties are satisfied with respect to $C$, then $G \in {\cal{H}}$ and hence $\mu(G)=2$. If any of the property is not satisfied with respect to $C$, $G \notin {\cal{H}'}=\cal{H}$ (from \cref{hish}) and hence $\mu(G)=1$. However, finding Hamiltonian cycles in general graphs is an NP-hard problem. Fortunately, since we are only interested in graphs where $\mu(G)=2$, we can give a simple efficient algorithm that returns a Hamiltonian cycle of $G$, using \cref{algo_proof_lemma}. 

We pick any vertex $v \in V(G)$, arbitrarily. Note that the maximum degree is at most $4$ when $\mu(G)=2$ from \cref{max_degree_4}. Therefore, we try to find a Hamiltonian cycle only when $d(v) \leq 4$. There are at most $\binom{4}{2}$ pairs of edges incident on $v$. Since $G$ is a matching covered graph, for each pair $\{e_1,e_2\}$, we can find arbitrary perfect matchings $M_1,M_2$ containing $e_1,e_2$ respectively, using any maximum matching algorithm \cite{DBLP:conf/focs/MicaliV80}. If their union is not a Hamiltonian cycle for all $\binom{4}{2}$ pairs, $\mu(G)=1$ graph from \cref{algo_proof_lemma}. If their union is a Hamiltonian cycle $C$ for at least one pair, then we check if $G \in {\cal{H}}$ or $G \notin {\cal{H}'}$ with respect to $C$. Recall that since $G$ is a matching covered graph, by \cref{2_characterization_result}, $\mu(G)=2$ if and only if $G \in {\cal{H}'}={\cal{H}}$.
\end{proof}

We note that one can also rephrase our result, i.e. \cref{2_characterization_result} in terms of SPQR trees and then use the linear time algorithm to find the SPQR tree of the matching covered graph $G$ as a black box to detect graphs with matching index $2$. However, the bottleneck is to find the matching covered graph $G=mcg(H)$ in both approaches, which takes $O(|V||E|)$ time. We present the current algorithm as it naturally follows from the structural classification.

\section{Edge-weighted version}\label{conclusion}
So far, we answered a natural graph theoretic question arising from quantum photonic experiments: What is the structure of the graphs which admit a PM-Valid $2$-edge colouring? In this section, we provide a little more information on graphs that arise from the quantum optical experiments in the general case; namely when we allow destructive interference. The interested reader can refer to \cite{krenn2019questions} and \cref{appendix} to know how these graphs correspond to the experiments.

Usually, in an edge colouring, each edge is associated with a natural number. In such edge colourings, the edges are assumed to be monochromatic. But to capture the full power of quantum photonic experiments that we are trying to model, we have to consider a generalized edge colouring using bi-chromatic edges, i.e. one half coloured by a certain colour and the other half coloured by a possibly different colour (as shown in \cref{fig:main_example}). So, we develop some new notation to describe bi-chromatic edges. 

One can imagine an edge $e=\{u,v\}$ to be of two halves, where the half containing $u$ is represented by $u_{h}(e)$ and the the half containing $v$ is represented by $v_{h}(e)$. An edge colouring $c$ associates a colour $i$ to $u_{h}(e)$ and a colour $j$ to $v_{h}(e)$ for some colours $i,j \in \mathbb{N}$ (as shown in \cref{fig:main_example}). For such an edge $e$, if $i\neq j$ we call $e$ to be a bi-chromatic edge and if $i=j$, we call $e$ to be a monochromatic edge. A weight assignment $w$ assigns every edge $e$ a weight $w(e) \in \mathbb{C} \setminus \{0\}$. (As a zero-weight edge $e$ is the same as the edge $e$ being absent, we do not need to consider zero-weight edges.)


\begin{definition}
The weight of a perfect matching $P$, $w(P)$ is the product of the weights of all its edges $\prod\limits_{e\in P}w(e)$
\end{definition}

\begin{definition}
The weight of a graph $G$, $w(G)=\sum_{P \in \cal{P}} w(P)$ where $\cal{P}$ is the set of all perfect matchings of $G$.
\end{definition}

A vertex colouring $vc$ associates a colour $i$ to each vertex in the graph where $i\in \mathbb{N}$. We use $vc(v)$ to denote the colour of vertex $v$ in the vertex colouring $vc$. 

\begin{definition}
For an edge-weighted edge-coloured graph $G_c^w$, each vertex colouring $vc$ naturally defines an \textit{edge-filtering} operation $\Theta_{vc}(G_c^w)=H_c^w$ as follows: $V(H_c^w)=V(G_c^w)$ and $e=\{u,v\} \in E(H_c^w)$ if and only if $e \in E(G_c^w)$ and $u_h(e),v_h(e)$ are coloured with $vc(u),vc(v)$ respectively.
\end{definition}
We say that $\Theta_{vc}(G_c^w)$ is the $vc$-filtered subgraph of  $G_c^w$.



\begin{definition}
The weight of a vertex colouring $vc$ with respect to $G_c^w$ is defined to be the weight of the subgraph of $G_c^w$ filtered by $vc$, $w(vc,G_c^w)=w(\Theta_{vc}(G_c^w))$
\end{definition}
Note that if $\Theta_{vc}(G_c^w)$ does not have a perfect matching, then $w(vc)=0$ by definition. When the graph $G_c^w$ is understood from the context, we will use $w(vc)$ to denote $w(vc,G_c^w)$.

\begin{definition}
An edge-weighted edge-coloured graph $G_c^w$ is said to be \textit{valid} with dimension $\bar\mu({G,c,w})$ (or $\bar\mu({G,c,w})$-valid), if:
\begin{enumerate}
    \item If there exists $\bar\mu({G,c,w})$ monochromatic vertex colourings with a weight of $1$.
    \item All other vertex colourings have a weight of $0$.
\end{enumerate}
\end{definition}

An example of a $2$-valid edge-coloured edge-weighted graph is shown in \cref{fig:main_example}

\begin{figure}[t!]
    \centering   
    \begin{minipage}{0.95\textwidth}
\centering    
{\includegraphics[width=63mm]{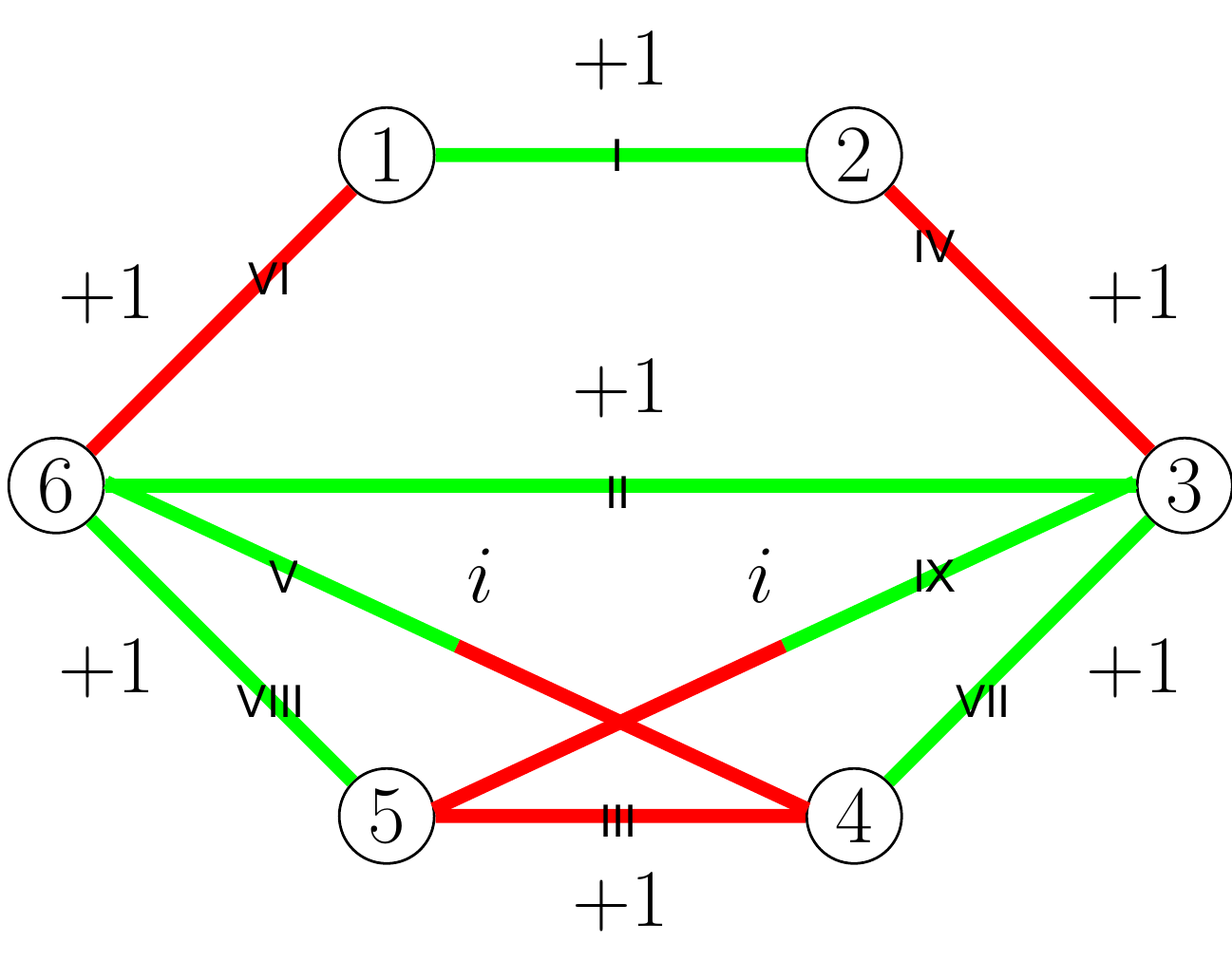}}
\caption{A $2$-valid edge-coloured edge-weighted graph $G_c^w$ (red corresponds to $0$ and green corresponds to $1$)
}
\label{fig:main_example}
    \end{minipage}
   \hspace*{1cm}
\end{figure}

\begin{figure}[t!]
    \centering
\centering    
\begin{subfigure}[b]{0.21\textwidth}
         \centering
         \fbox{\includegraphics[width=0.9\textwidth]{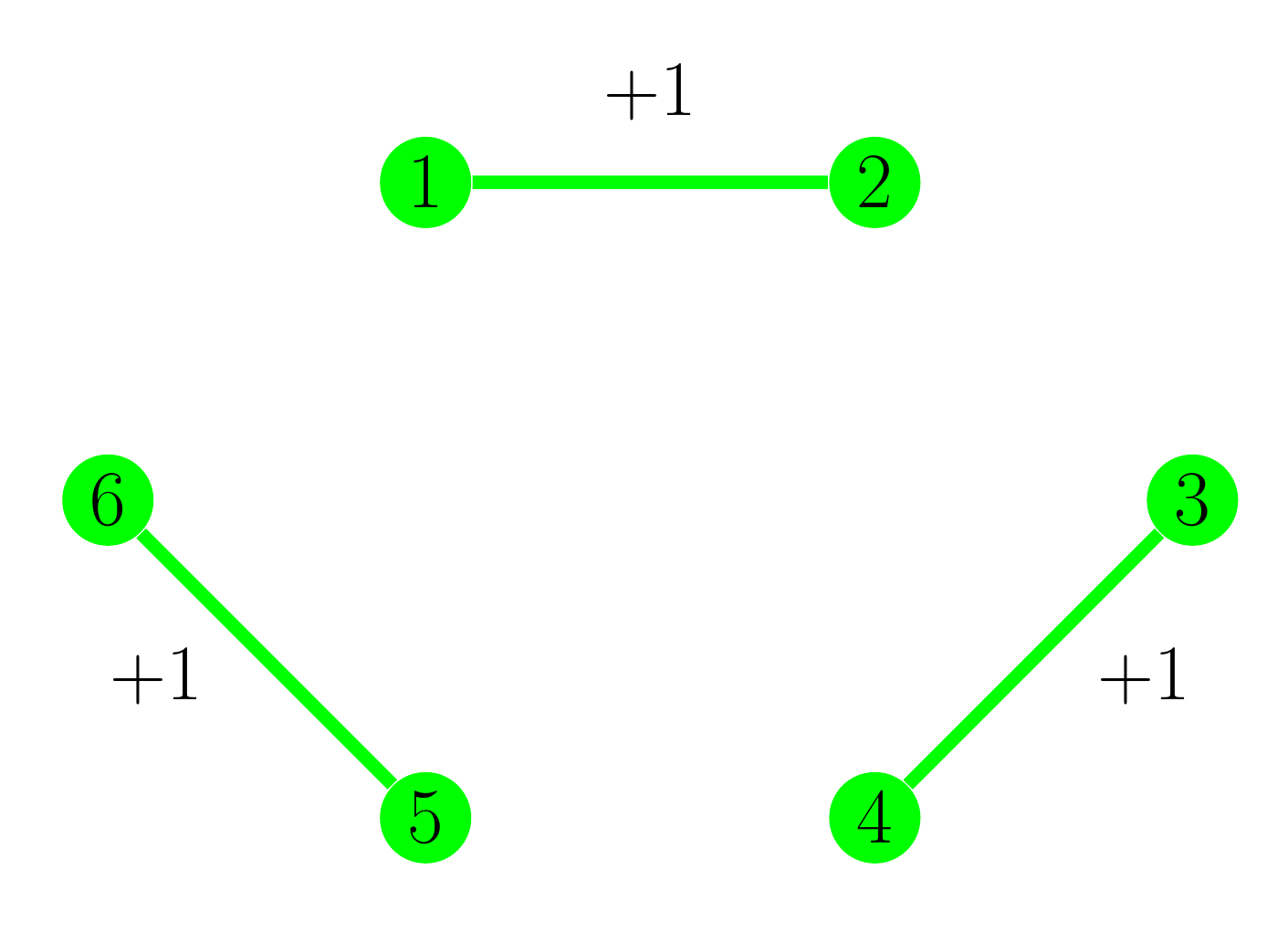}}
         \caption{$\Theta_{|111111\rangle}(G_c^w)$}
         \label{fig:pm1}
\end{subfigure}
\hfill
\begin{subfigure}[b]{0.21\textwidth}
         \centering
         \fbox{\includegraphics[width=0.9\textwidth]{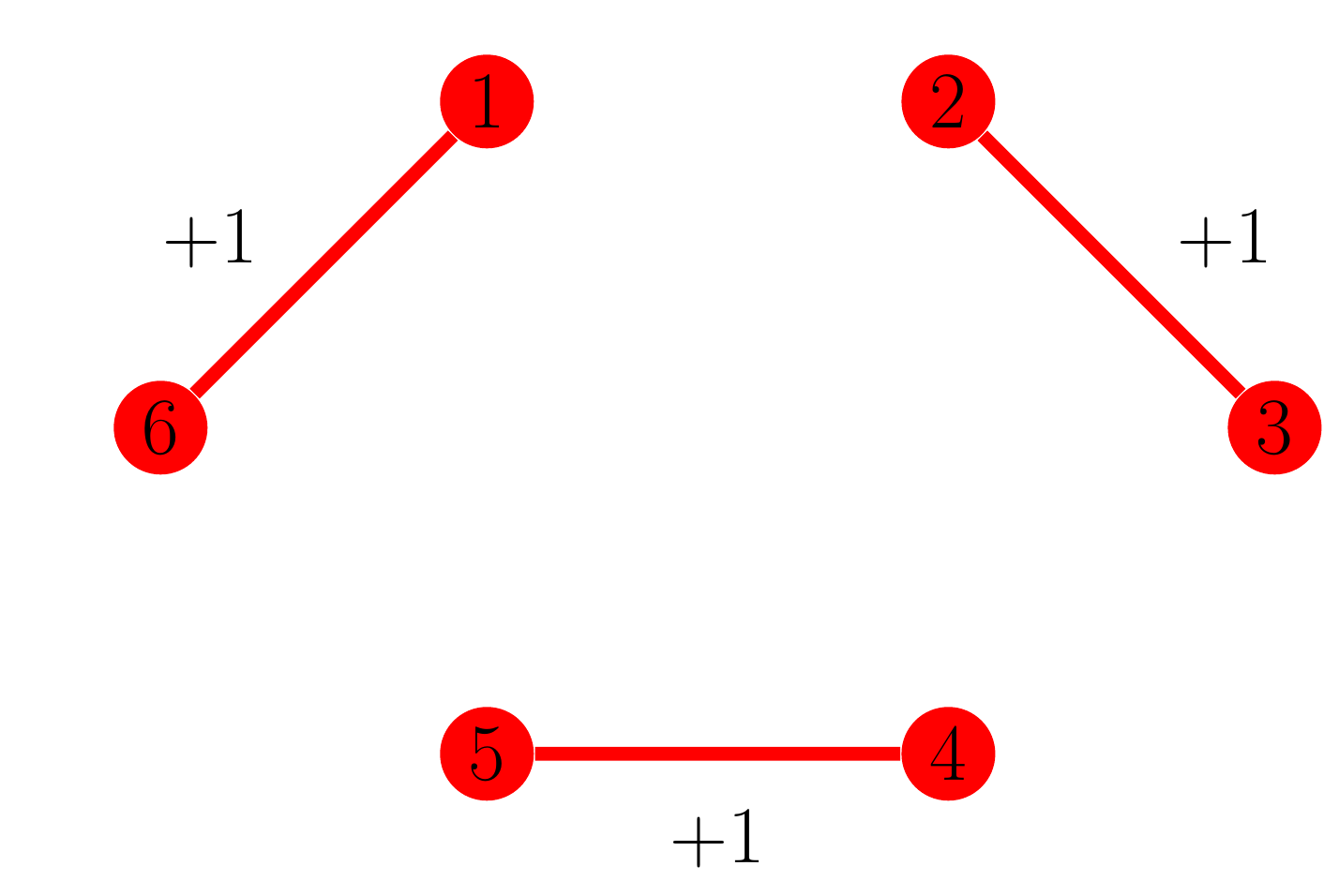}}
         \caption{$\Theta_{|000000\rangle}(G_c^w)$}
         \label{fig:pm2}
\end{subfigure}
\hfill
\begin{subfigure}[b]{0.21\textwidth}
         \centering
         \fbox{\includegraphics[width=0.9\textwidth]{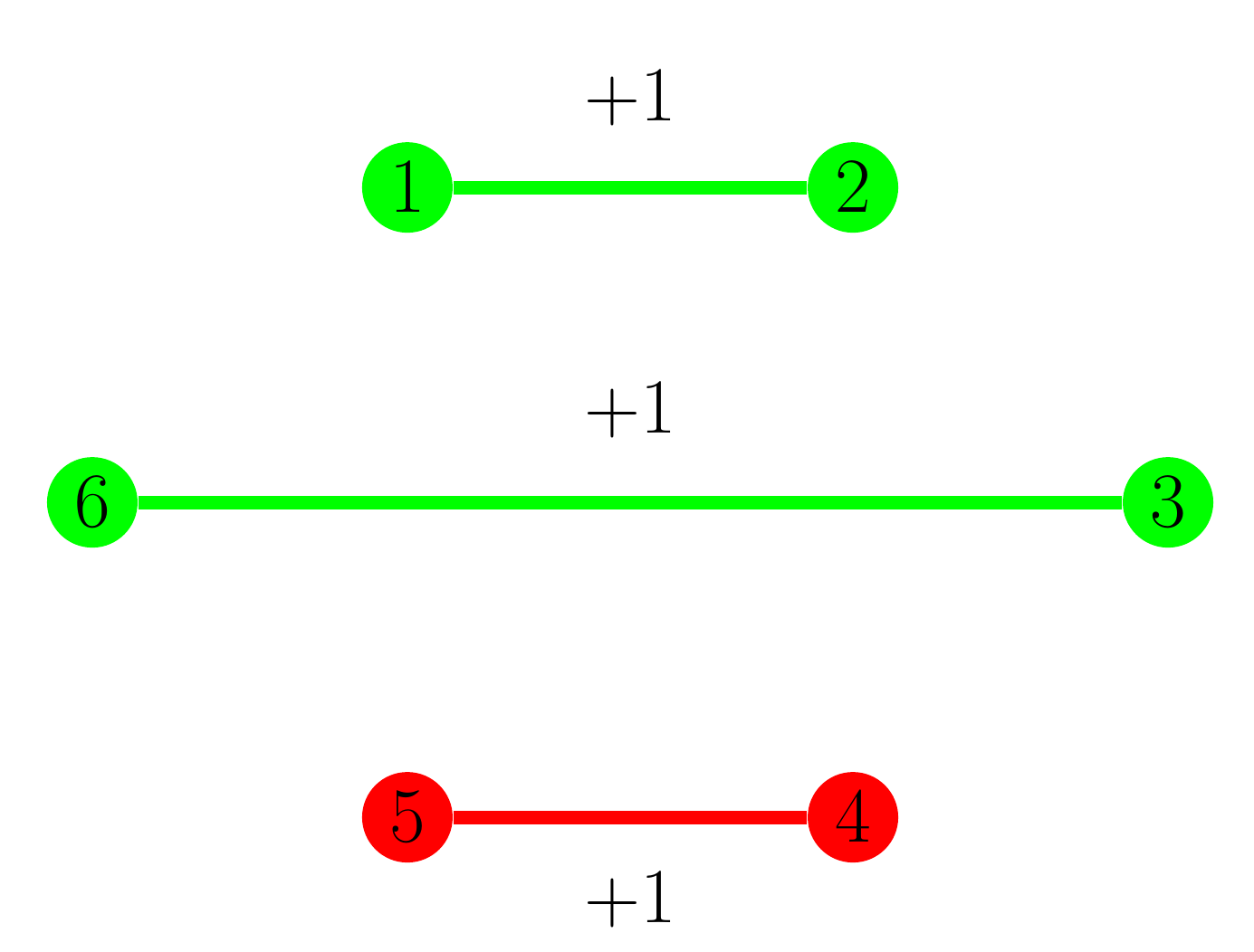}}
         \caption{First perfect matching of $\Theta_{|111001\rangle}(G_c^w)$}
         \label{fig:pm3}
\end{subfigure}
\hfill
\begin{subfigure}[b]{0.21\textwidth}
         \centering
         \fbox{\includegraphics[width=0.9\textwidth]{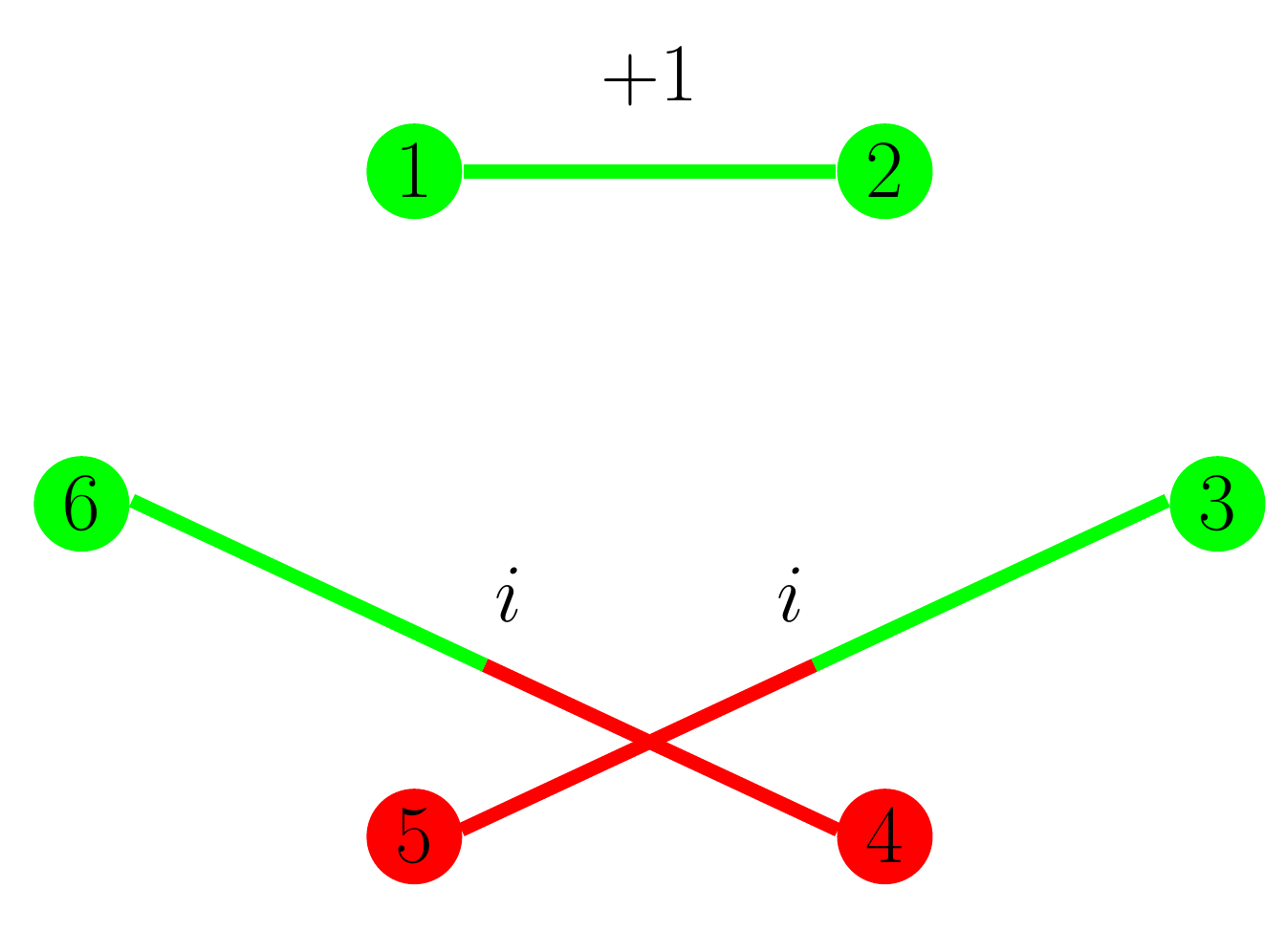}}
         \caption{Second perfect matching of $\Theta_{|111001\rangle}(G_c^w)$}
         \label{fig:pm4}
\end{subfigure}

\caption{The vertex colourings induced by the perfect matchings of \cref{fig:main_example}.}
\label{fig:perfect_matchings}
\end{figure}

\begin{definition}
For a graph $G$, the maximum $k$ such that there exists an edge colouring $c$ and edge-weight assignment $w$ for which $G_c^w$ is $k$-valid (valid with dimension $k$) is defined to be $\Bar{\mu}(G)$. 
\end{definition}

In a graph $G_c^w$, for the weight of the non-monochromatic vertex colouring $vc$ to be zero, it is not necessary that $\Theta_{vc}(G_c^w)$ does not have a perfect matching; the weights of perfect matchings can cancel out each other. For instance, perfect matchings in \cref{fig:pm3} and \cref{fig:pm4} add up to zero. However, when all edge weights are positive real numbers, such cancellations are not possible and hence there cannot be any non-monochromatic perfect matchings in $G_c^w$. Therefore, when all the weights are restricted to positive real values, the weighted problem would reduce to the unweighted problem of PM-valid edge colourings which was discussed in the earlier sections. It is easy to see that $\bar{\mu}(G)\geq \mu(G)$ (\cref{weight_unweight}). 

As we have observed already in terms of weights, the difference between unweighted and weighted case is that essentially we are not restricted to only real positive weights but are allowed to take any complex weights. In quantum physics terms this is expressed by saying that destructive interference is allowed. Krenn and Gu made a conjecture that allowing destructive interference does not help to get a higher dimension (and hence a statement analogous to \cref{bogdanov} for edge-weighted edge-coloured graphs will hold). The reader not familiar with quantum physics can safely discard these comments on destructive interference as the problem and its solution can be understood from a graph theoretic point of view. 


\begin{conjecture}[Krenn-Gu Conjecture]
\label{KGconj}
For a multigraph $G$ which is non-isomorphic to $K_4$, $\Bar{\mu}(G)\leq 2$ and $\Bar{\mu}(K_4)=3$.   
\end{conjecture}

Surprisingly, we prove that once the structure corresponds to $\mu(G)=2$, then the freedom to choose any colouring (with possibly bi-chromatic edges) and to choose complex weight functions does not help to achieve a higher value of $\Bar\mu(G)$ for simple graphs (i.e. destructive interference indeed doesn't help in this case as Krenn and Gu predicted for all cases). 
\begin{theorem}
\label{type_2_resolution}
For a simple graph $G$, if $\mu(G)\neq 1$, then $\bar\mu(G)=\mu(G)$.
\end{theorem}
The above theorem is an application of our main result, which is the structural classification of graphs with $\mu(G)=2$. We first observe that a certain gadget is present in the graphs with $\mu(G)=2$, from \cref{2_characterization_result}. We then prove that this gadget would prevent from achieving a higher value of $\Bar{\mu}(G)$ in \cref{sec:proof_type_2_resolution}. A natural question one can ask is if  $\mu(G)$ and $\bar\mu(G)$ should always be the same. We prove that this is not the case in \cref{sec:mu1_mubar2}.
\begin{observation}\label{mu1_mubar2}
There exists a graph $G$, with $\mu(G) = 2$ and $\bar\mu(G)=1$.
\end{observation}


\section{Proof of \cref{type_2_resolution}: Weighted case}
\label{sec:proof_type_2_resolution}
\subsection{Preliminaries for weighted edge-coloured graphs}
\label{sec:prelim}
We first make some simple observations for the weighted version. Let $G_c^w$ be a valid edge-coloured edge-weighted graph. 
Note that a colouring $c$ and a weight assignment $w$ of $G$ induces a colouring and a weight assignment for every subgraph of $G$, respectively. When there is no scope for confusion, we use $c,w$ itself to denote this induced colouring and weight assignment, respectively. It is easy to see that if $c,w$ makes $G$ a $k$-valid edge-coloured edge-weighted graph, they also make $mcg(G)$ a $k$-valid edge-coloured edge-weighted graph and $\bar{\mu}(G,c,w)=\bar{\mu}(mcg(G),c,w)$. Therefore, $\bar{\mu}(G)=\bar{\mu}(mcg(G))$.

For $S\subseteq V(G)$, we define $\mathbf{i}_S$ to be the monochromatic vertex colouring using colour $i$ over the induced subgraph $G_c^w[S]$. The weight of $\mathbf{i}_S$, denoted by $w(\mathbf{i}_S)=w(\Theta_{\mathbf{i}_S}(G_c^w[S]))$. We emphasize that this weight is with respect to the induced subgraph $G_c^w[S]$. Let $G_c^w[\mathbf{i},S]$ denote the subgraph of $G_c^w[S]$ formed by all monochromatic edges of $G_c^w[S]$ of colour $i$. 
\begin{observation}
\label{partition_split_lemma}
For $S \subseteq V(G)$, let the connected components of $G_c^w[\mathbf{i},S]$ be $\bigsqcup\limits_{j=1}^{r} S_j = S$. Then the weight of its monochromatic vertex colouring $\mathbf{i}_S$ is
$$w(\mathbf{i}_S) =
\prod\limits_{j=1}^{r}w(\mathbf{i}_{S_j})$$
\end{observation}

\begin{observation}\label{min_degree_bound}
$\bar{\mu}(G)\leq \delta(G)$
\end{observation}

\begin{proof}
 Let $c, w$ be such that $\bar{\mu}(G, c, w) = \bar{\mu}(G)$. Note that in $G_c^w$ , there are $ \bar{\mu}(G)$ colour classes containing a monochromatic perfect matching. Therefore, any vertex $v$ of $G$ has at least one monochromatic edge incident on it from each of the $ \bar{\mu}(G)$ colours. Therefore, v has degree at least $ \bar{\mu}(G)$. Therefore, the minimum degree of $G$ is at least $ \bar{\mu}(G)$.
\end{proof}

\begin{observation}\label{weight_unweight}
$\bar{\mu}(G)\geq \mu(G)$
\end{observation}

\begin{proof}
It is sufficient to prove that given an unweighted edge-coloured graph $G_c$ which is PMValid, it is possible to assign weights $w$ to the edges so that the graph $G_c^w$ $\mu(G)$-valid. We first note that any edge which is part of a perfect matching in $G_c$ cannot be bi-chromatic as $c$ is PMValid. Let the colour of an edge $e$ be represented by $c(e)$. Let $pm_{c(e)}>0$ denote the number of monochromatic perfect matchings of colour $c(e)$ in $G_c$. Consider the weight function $w(e)= (pm_{c(e)})^{-2/n}$ on $G_c$. It is easy to see that the weight of each monochromatic perfect matching of colour $c(e)$ is ${pm_{c(e)}}^{-1}$ as each perfect matching has $n/2$ edges. Since there are $pm_{c(e)}$ many monochromatic perfect matchings of colour $c(e)$, the weight of all $\mu(G)$ monochromatic vertex colourings is $1$. As all perfect matchings are monochromatic, all non-monochromatic vertex colourings have weight $0$. Therefore, $G_c^w$ is $\mu(G)$-valid.
\end{proof}

\subsection{Proof of \cref{type_2_resolution}: If $\mu(G)\neq 1$, then $\bar\mu(G)=\mu(G)$}

\begin{proof}
If $\mu(G)=0$, from \cref{simple_classification_theorem}, $G$ has no perfect matching. Therefore, $\bar\mu(G)=0=\mu(G)$. If $\mu(G)=3$, we know that $G$ is $K_4$. As $K_4$ has $6$ edges, it has at most $3$ disjoint perfect matchings. Hence $\bar\mu(G)\leq 3$. We also know that $\bar\mu(G) \geq \mu(G) =3$. Therefore $\bar\mu(K_4)=3=\mu(K_4)$. We now prove that if $\mu(G) = 2$, then $\bar\mu(G)=2$ 

If $\mu(G)=2$, then for $mcg(G)$, from \cref{4_cycle_observation} there exists a drum $D$ (with respect to some Hamiltonian cycle) such that all vertices of the smaller part, ${\cal{P}}(D)$ have degree at most $2$. If ${\cal{P}}(D)$ is non-empty, then $\delta(mcg(G)) \leq 2$ and from \cref{min_degree_bound}, $\bar\mu(G) \leq 2$. 

If ${\cal{P}}(D)$ is empty, then the strap between two vertices, say $a,c$ of the drum $D$ is just the cycle edge $ac$. Now, it is easy to see that $a,c$ have degree $3$. Let the other two vertices from the drum $D$ be $b,d$ such that $abcd$ is a $4$ cycle as shown in \cref{fig:gadget}. We now prove that existence of a slightly more general structure (of which \cref{fig:gadget} is a special case) would imply that $\bar\mu(G) \leq 2$ in \cref{4_cycle_lemma}.

\begin{figure}
   \centering   
    \begin{minipage}{0.49\textwidth}
\centering    
{\includegraphics[width=.4\linewidth]{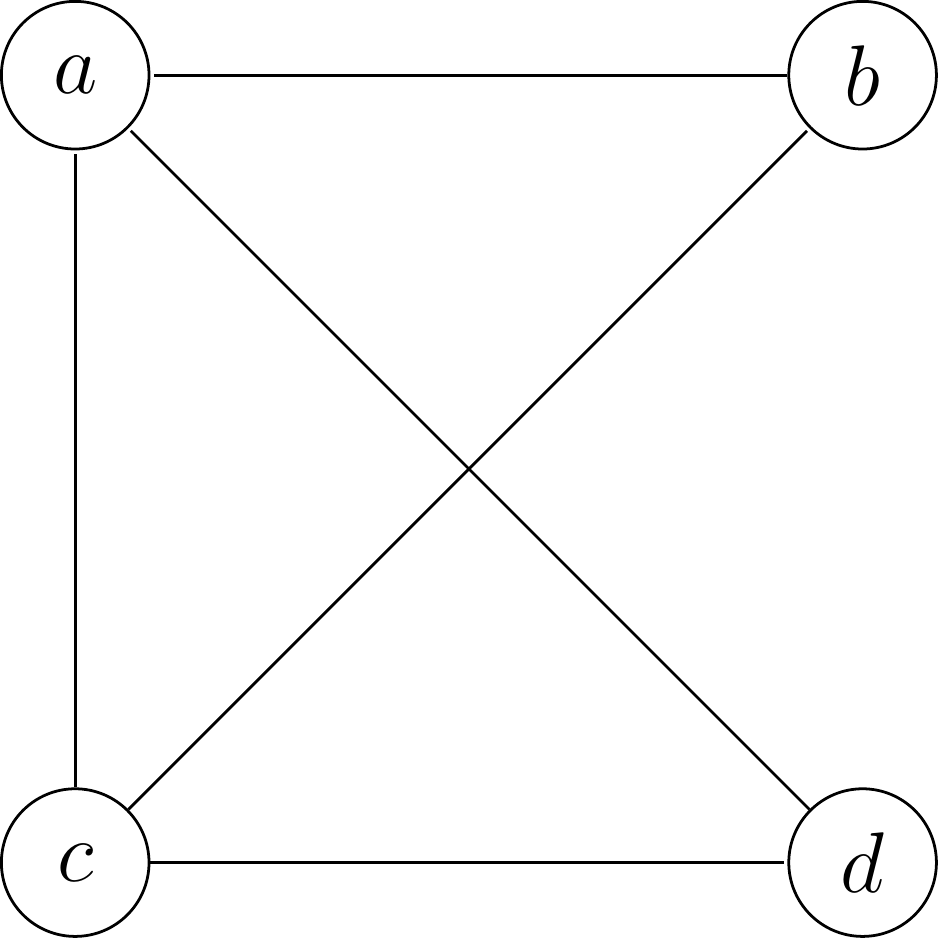}}
\caption{Gadget for \cref{4_cycle_lemma}}
\label{fig:gadget}
    \end{minipage}
    \begin{minipage}{0.49\textwidth}
\centering    
{\includegraphics[width=.69\linewidth]{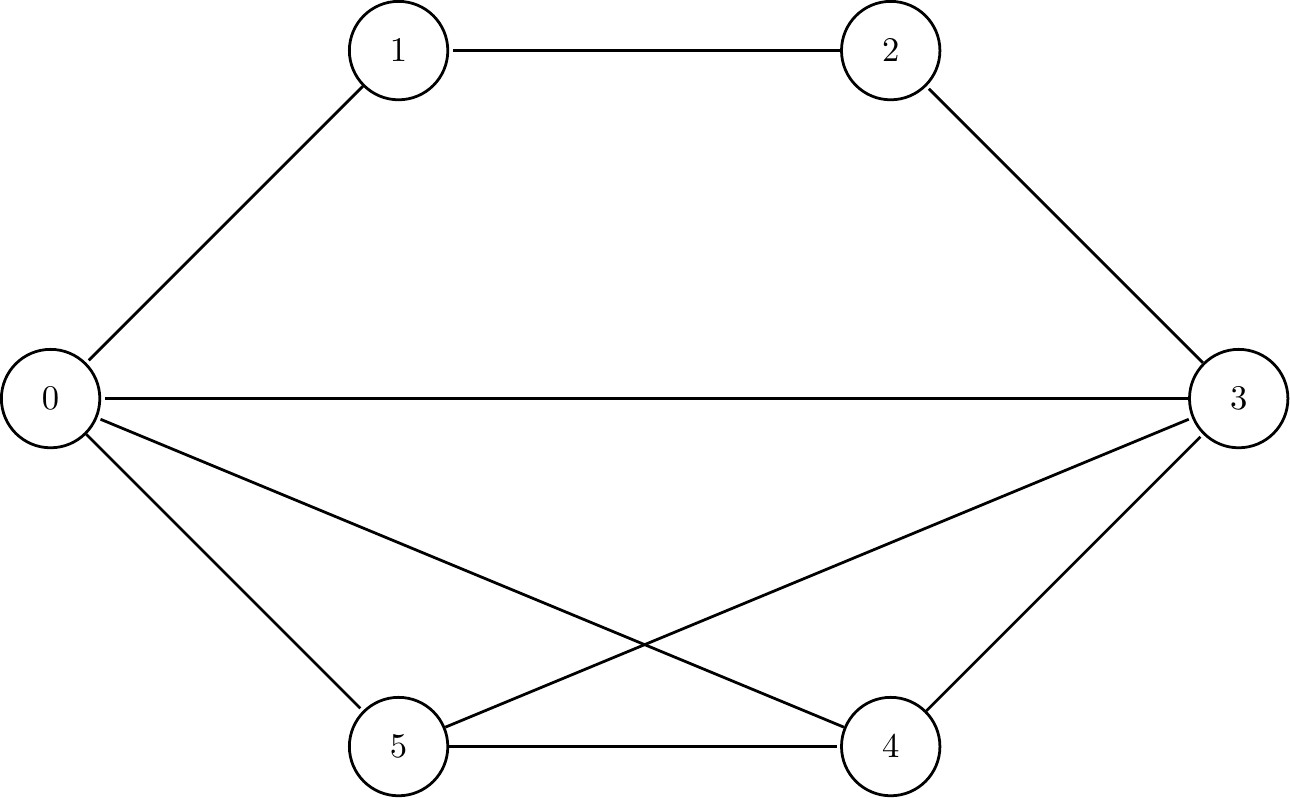}}
\caption{Example for \cref{mu1_mubar2}: \\ Graph $G$ with $\bar\mu(G)=1$ and $\mu(G)=2$}
\label{fig:uncoloured_example}  
    \end{minipage}
\end{figure}

\begin{lemma}
\label{4_cycle_lemma}
If there exists $a,b,c,d \in V(G)$ such that $d(a)=d(c)=3$ and $abcd$ form a $4$ cycle, then $\bar\mu(G) \leq 2$.
\end{lemma}
\begin{proof}
Towards a contradiction, let $\bar\mu(G) \geq 3$. As there is a $3$ degree vertex, $\delta(G)\leq 3$. From \cref{min_degree_bound}, $\bar\mu(G) \leq \delta(G)\leq 3$. Therefore, $\bar\mu(G) = 3$ and hence there exists a colouring $c$ and weight $w$, such that $G_c^w$ is $3$-valid. Let the three colour classes be $1,2,3$. Note that as $a,c$ have degree $3$, they should have at least one monochromatic edge incident on them from each of the colour classes $1,2,3$. Therefore, they can neither have two edges from the same colour class incident on them or a bi-chromatic coloured edge incident on them. It is easy to see that all the edges of the cycle $abcd$ are monochromatic.

Also note that both the edges incident on $b$ from the cycle $abcd$ cannot belong to the same colour class. This is because if $ab,bc$ (which are monochromatic) have the same colour, say $1$, a monochromatic perfect matching associated with the colour $1$ must contain both $ab$ and $bc$ (as the only one $1$-coloured monochromatic edge incident on $a$ is $ab$ and the only one $1$-coloured monochromatic edge incident on $c$ is $bc$). But there can not be two edges from a perfect matching incident on $b$. Same observation holds for the vertex $d$, that is $da,dc$ must have different colours.

Since $3$ colour classes contain the $4$ edges of the cycle $abcd$, a colour class must contain at least two edges. As these edges can not be adjacent as proved above, without loss of generality, let $ab,cd$ be of colour $1$. Therefore, each of $bc$ and $ad$ can not be coloured $1$; they are coloured with $2$ or $3$. Let $S = V(G)-\{a,b,c,d\}$. From \cref{partition_split_lemma}, $w(\mathbf{1}_{V(G)})=w(\mathbf{1}_{S})w(ab)w(cd)=1$. Therefore, $w(\mathbf{1}_{S}) \neq 0$.

Consider the vertex colouring $vc$, where $S$ is coloured $1$; $\{b,c\}$ are coloured $c_{bc}$; and $\{a,d\}$ is coloured $c_{ad}$. Recall that $c_{bc}, c_{ad}$ cannot be equal to $1$ (as they are adjacent to the edge $ab$ which is of colour $1$). As  $a,c$ have three monochromatic edges of different colours incident on each them, the only edge of colour $c_{ad}$ incident on $a$ is $ad$. Similarly, the only edge of colour $c_{bc}$ incident on $c$ is $bc$. Therefore any perfect matching of $\Theta_{vc}(G_c^w)$ must contain the edges $bc,ad$. Hence, it is easy to see that $w(vc)=w(\mathbf{1}_S)w(bc)w(ad)=0$. By assumption as the edge weights are non-zero, we get $w(\mathbf{1}_S)= 0$. But we know that $w(\mathbf{1}_{S}) \neq 0$. Contradiction. 
 
\end{proof}
We know that  $\bar\mu(G) \geq \mu(G) =2$. Therefore, $\bar\mu(G) = 2$.
\end{proof}

\subsection{Proof of \cref{mu1_mubar2}: A graph $G$, with $\mu(G) = 1$ and $\bar\mu(G)=2$.}
\label{sec:mu1_mubar2}

\begin{proof}
We prove that the graph $G$ shown in \cref{fig:uncoloured_example} satisfies $\mu(G)=1$ and $\bar\mu(G)=2$. Consider the cyclic ordering along a Hamiltonian cycle $C$ as shown. As there is an illegal edge $(3,6)$ in $G$ with respect to $C$, $G \notin {\cal{H}}$ by defintion of ${\cal{H}}$. It follows from \cref{2_characterization_result} that  $\mu(G)\neq 2$ and hence $\mu(G) = 1$ from \cref{bogdanov}.

There is a $c,w$ on $G$ as shown in \cref{fig:main_example} such that $\bar\mu(G,c,w)=2$. 
Hence $\bar\mu(G)\geq 2$. Note that $\delta(G)=2$. Since there exits a vertex of degree $2$, it is easy to see that  $\bar\mu(G)\leq 2$. Therefore, $\bar\mu(G) = 2$.

\end{proof}

	\bibliographystyle{apalike}
	\bibliography{references.bib}


\section{Appendix: The experiment for GHZ state generation}\label{appendix}
The material discussed here is not a contribution of this paper. It is added for the convenience of the reader. We will first outline the experimental setup for GHZ state generation, skipping some technical details. The reader may refer to \cite{krenn2019questions} for the exact details. Several crystals (shaded in blue in \cref{fig:main_experiment}) are used that are intercepted by laser beams. This results in the generation of an entangled photon pair in each crystal. An entangled photon pair is associated with an amplitude (which is a complex number), and each photon in the pair is associated with a mode number. The parameters of the laser can be adjusted so that each photon in a pair gets the desired mode number and the photon pair gets the desired amplitude. It is convenient to represent each possible mode number with a distinct colour. For instance, in crystal V in \cref{fig:main_experiment}, one photon got the colour red (mode number $0$), and the other photon got the colour green (mode number $1$), and the amplitude of the pair is the complex number $i$. Each photon has a photon path associated with it, which ends at a photon detector (the red photon in crystal V is associated with a photon path ending at photon detector $4$ and the green photon in crystal V is associated with a photon path ending at photon detector $6$). Note that the photon paths of multiple photons (from different crystals) can end at the same photon detector. But each of these photon detectors can detect exactly one photon. As the photons come in entangled pairs, it so happens that they are always detected in pairs, i.e. if the photon detector $6$ detects the green photon in crystal $V$, then the red photon in crystal $V$ must have been detected by the photon detector $4$. Thus we can talk about an entangled photon pair being detected rather than an individual photon being detected. It is worth noting that the amplitude of a photon pair is related to the probability of it being detected. 

\begin{figure}[t!]
    \centering   
    \begin{minipage}{0.99\textwidth}
\centering    
{\includegraphics[width=90mm]{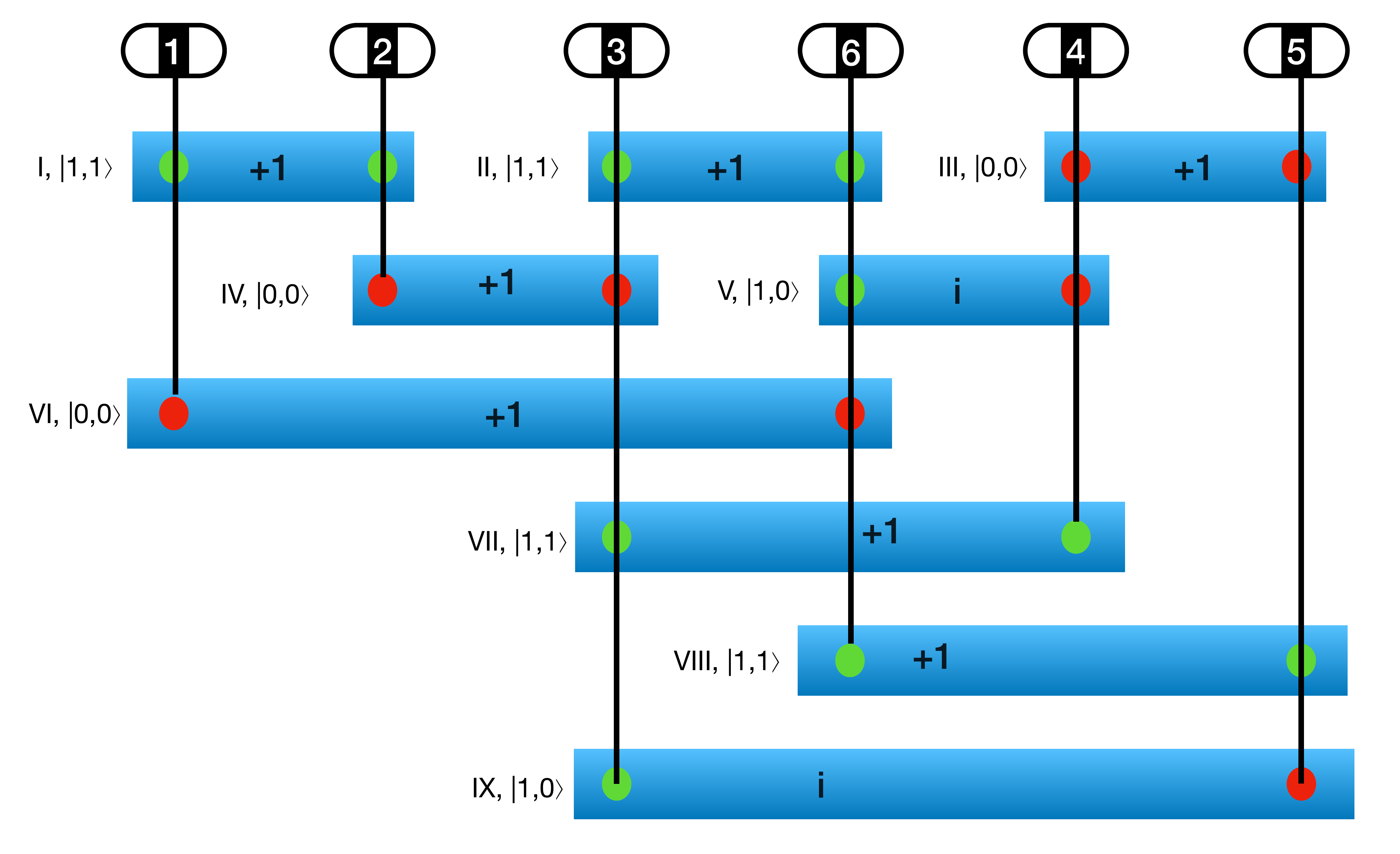}} \caption{An experiment to produce a $6$ particle $2$-dimensional GHZ state.}
 \label{fig:main_experiment}
    \end{minipage}
\end{figure}

In a measurement, one can know the mode numbers of the photons detected on each photon detector. Using ket notation, we can represent the result of the measurement as $|x_1x_2\ldots x_n\rangle$, where $x_i$ denotes the mode number of the photon detected by the $i$th photon detector. For instance, the only way the measurement can result in the multi-photonic term $|11111\rangle$ is if the photon pairs from the crystals $\{\text{I,VII,VIII}\}$ are detected. Whereas the multi-photonic term $|111001\rangle$ can result due to two events, namely, if the photon pairs from the crystals $\{\text{I,V,IX}\}$ are detected or if the photon pairs from the crystals $\{\text{I,II,III}\}$ are detected. The weight of an event (which is related to the probability of occurrence of all photon pairs in that event being simultaneously detected) is the product of the amplitudes of the detected photon pairs. For instance, the weight of the event $\{\text{I,VII,VIII}\}$ is $1\cdot1\cdot1=1$. Now the weight of a measurement resulting in the multi-photonic term is the sum of the weights of the event causing it. For instance, the term $|111001\rangle$ is caused by the events $\{\text{I,V,IX}\},\{\text{I,II,III}\}$ and hence the weight of $|111001\rangle$ is $1\cdot i \cdot i+ 1\cdot1\cdot1=0$. A multi-photonic term is pure if all the photon detectors detect the same mode number and are impure otherwise. For the experiment to lead to a GHZ state, the weights of impure multi-photonic terms should be zero, and pure multi-photonic terms should be non-zero and equal (this can be assumed to be one by normalizing). The number of non-zero multi-photonic terms corresponds to the dimension of the GHZ state.

In \cref{fig:main_experiment}, the event  $\{\text{I,VII,VIII}\}$ is causing the pure multi-photonic term $|11111\rangle$ and the weight of $|11111\rangle$ is $1\cdot1\cdot1=1$ as required. The event $\{\text{III,IV,VI}\}$ is causing the pure term $|000000\rangle$ and the weight of $|000000\rangle$ is $1\cdot1\cdot1=1$ as required. The events $\{\text{I,V,IX}\},\{\text{I,II,III}\}$ are causing the impure term $|111001\rangle$ and the weight of $|111001\rangle$ is $1\cdot i \cdot i+ 1\cdot1\cdot1=0$ as required. These are the only possible multi-photonic terms that can be caused by some event. Therefore, this experiment would give a $6$-particle $2$-dimensional GHZ state. $$|GHZ_{6,2}\rangle = \frac{1}{\sqrt{2}}\left(|000000\rangle + |111111\rangle \right)$$

\subsection{Connection to graph theory}
We now translate the experiment in \cref{fig:main_experiment} to the edge-weighted edge-coloured graph of \cref{fig:main_example}. Each photon detector corresponds to a vertex, and each photon pair generated in a crystal corresponds to an edge. If the photon paths through the photons in a crystal end at photon detectors $u,v$, then we put an edge between the vertices corresponding to the photon detectors $u,v$. The weight of an edge is the amplitude of the corresponding photon pair. Colours on the edges of the experiment graph capture the mode number. As the two photons in a pair can have different mode numbers, each edge can be bi-chromatic; half of an edge gets one colour, and the other half might get a different colour~\cite{krenn2019questions}. {For example, in \cref{fig:main_example}, there is a bi-chromatic edge between the vertices $3,5$ due to the photon pair generated in crystal IX. The photon whose path ends at the photon detector corresponding to vertex $3$ has a mode number corresponding to the colour green, and the other photon whose path ends at the photon detector corresponding to vertex $5$ has a mode number corresponding to the colour red.}

It is easy to see that a multi-photonic resulting from a measurement corresponds to a vertex colouring of the graph. 
The weight of a multi-photonic term (vertex colouring) should be the sum of the weights of all events causing it. An event corresponds to a set of photon pairs (disjoint edges) that cover all photon detectors (vertices) in the experiment (graph). It is easy to see that such an edge set is a perfect matching of the graph. Each perfect matching of an edge-coloured graph induces a vertex colouring, where each vertex of the graph has the colour of the unique half edge of the matching incident to the vertex; see, for instance, \cref{fig:perfect_matchings}. The weight of a perfect matching (event) is the product of the weights of all its edges (amplitudes of the crystals). A vertex colouring is defined to be feasible if it is induced by at least one perfect matching. It is now easy to see that the weight of a feasible vertex colouring is the sum of weights of all perfect matchings inducing that vertex colouring. Therefore, in graph theoretic terms, the quantum state produced by an experiment is a GHZ state when
\begin{enumerate}
    \item All feasible monochromatic vertex colourings have a weight of $1$.
    \item All non-monochromatic vertex colourings have a weight of $0$.
\end{enumerate}

We consider the weight of all vertex colourings, which are not feasible, to be zero by default. If the monochromatic vertex colouring of colour $i$ is not {feasible}, then all edges with at least half of it coloured $i$ can be discarded as such a mode number $i$ would not help in increasing the dimension of the corresponding GHZ state. This is consistent with interpreting the weight of vertex colourings as the probability of the corresponding multi-photonic term. So, we assume all monochromatic vertex colourings to be feasible when there is no scope for confusion. An edge-coloured edge-weighted graph in which the above two properties are satisfied is called a \textit{perfectly monochromatic} graph. 

As the experiment designer can set the mode numbers and amplitudes, if we find a colouring and a weight assignment for edges of a given finite undirected graph such that the resulting edge-coloured edge-weighted graph is {perfectly monochromatic}, then we can create an experiment which produces a GHZ state.

Experiments with parameters corresponding to a perfectly monochromatic graph produce a GHZ state. Its dimension and the number of particles are equal to the number of colours and vertices of the graph, respectively. The dimension for an edge-colouring $c$ and an edge weight assignment $w$, which makes the graph $G$ perfectly monochromatic, is represented as $\mu(G,c,w)$. In \cref{fig:3dimk4}, we have an edge-coloured edge-weighted perfectly monochromatic $K_4$. The corresponding experiment would produce a GHZ state of dimension $3$ using $4$ particles. 


For a given graph $G$, many possible colourings and weight assignments may make it perfectly monochromatic. For each such perfectly monochromatic edge-coloured edge-weighted graph, a dimension is achieved. The maximum dimension achieved over all possible perfectly monochromatic edge-coloured edge-weighted graphs with the finite undirected graph $G$ as their skeleton is known as the \textit{matching index} of $G$, denoted by $\mu(G)$. 

\begin{figure}[t!]
    \centering
         \centering
         {\includegraphics[width=40mm]{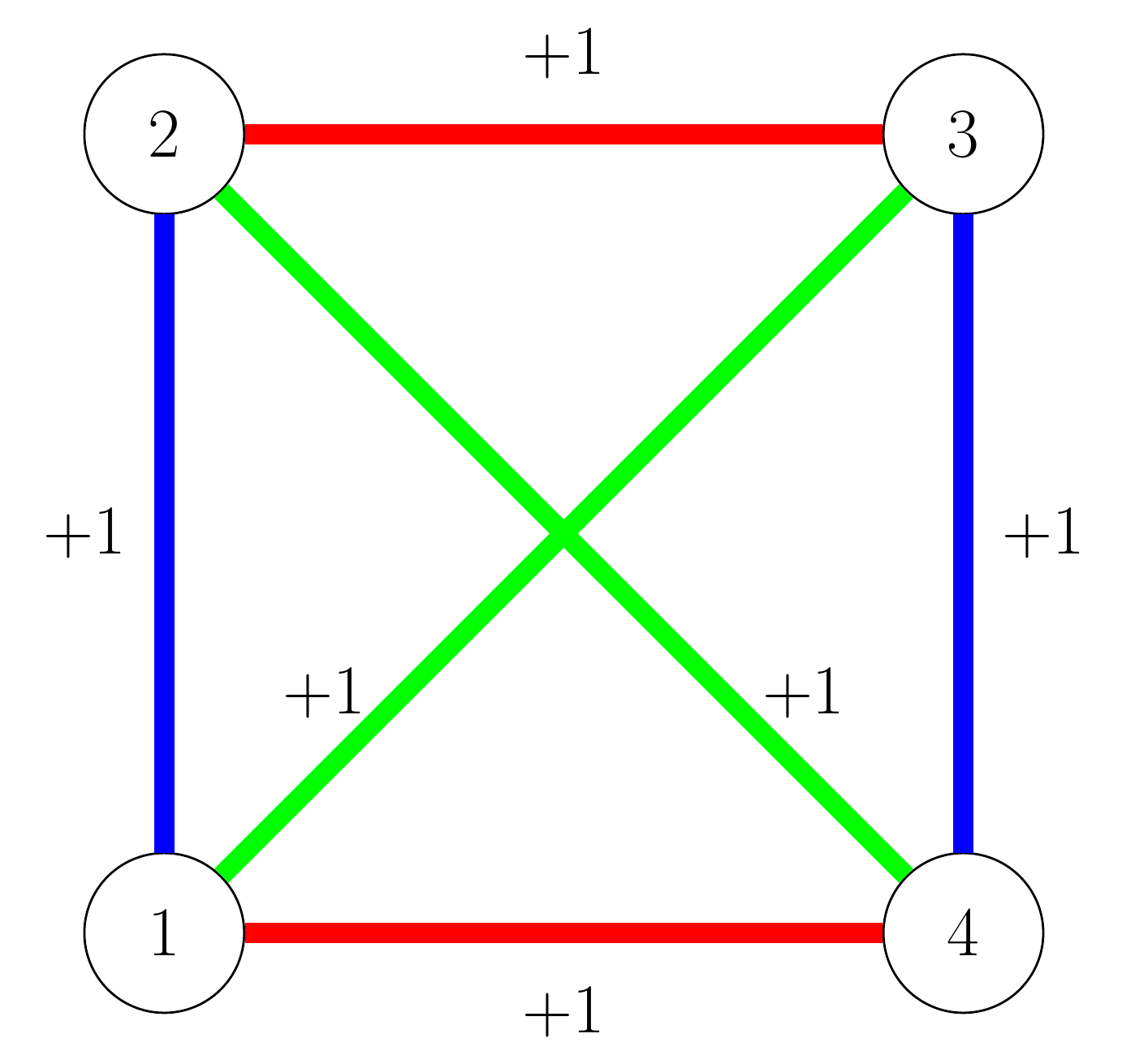}}
         \caption{$K_4$ with $3$ dimensions. 
         }
         \label{fig:3dimk4}

\label{fig:perfectly_mono_examples}
\end{figure}

It is easy to see that if a graph has a perfect matching, it must contain an even number of vertices. So, we consider matching indices of graphs with even and at least $4$ vertices for the rest of the manuscript. From \cref{fig:3dimk4}, we know that $\mu(K_4)\geq 3$ and, despite the use of huge computational resources \cite{AI1,AI2,neugebauer}, this is the only (up to an isomorphism) known graph of the matching index at least $3$. Any graph with a matching index of at least $3$ and $n$ vertices would lead to a new GHZ state of dimension at least $3$ with $n$ entangled particles. Motivated by this, this problem has been extensively promoted\cite{mixon_website,krenn_website}. Krenn and Gu conjectured that 
\begin{conjecture}
\label{krenn_conjecture}
If $G$ is non-isomorphic to $K_4$, then $\mu(G) \leq 2$ and $\mu(K_4) = 3$.
\end{conjecture}
Several cash rewards were also announced for a resolution of this conjecture \cite{krenn_website}. We note the following implications of resolving this conjecture
\begin{enumerate}
    \item Finding a counterexample for this conjecture would uncover new peculiar quantum interference effects of a multi-photonic quantum system using which we can create new GHZ states
    \item
    \begin{enumerate}
    \item Proving this conjecture would immediately lead to new insights into resource theory in quantum optics
    \item Proving this conjecture for different graph classes would help us understand the properties of a counterexample and guide experimentalists in finding it. This is particularly important since huge computational efforts are being put \cite{AI1,AI2,neugebauer}. 
    \end{enumerate}
\end{enumerate}

\end{document}